\documentclass[journal,twocolumn,10pt]{IEEEtran}
\usepackage{amsmath,amsfonts}
\usepackage[ruled,linesnumbered]{algorithm2e}
\usepackage{amsthm}
\usepackage{array}
\usepackage[caption=false,font=normalsize,labelfont=sf,textfont=sf]{subfig}
\usepackage{textcomp}
\usepackage{stfloats}
\usepackage{bm}
\usepackage{url}
\usepackage{verbatim}
\usepackage{graphicx}
\usepackage{amssymb}
\usepackage{cite}
\usepackage{arydshln}
\usepackage[colorlinks, citecolor=black, linkcolor=black]{hyperref}
\newcommand{\overbar}[1]{\mkern 2mu\overline{\mkern-2mu#1\mkern-2mu}\mkern 2mu}

\usepackage{xcolor}

\newtheorem{proposition}{Proposition}

\newtheorem{definition}{Definition}

\newcommand{\mb}[1]{\mathbf{#1}}
\newcommand{\mc}[1]{\mathcal{#1}}
\newcommand{\mbb}[1]{\mathbb{#1}}

\newcommand{\wt}[1]{\widetilde{#1}}
\newcommand{\wh}[1]{\widehat{#1}}
\newcommand{\ob}[1]{\overbar{#1}}

\begin{document}

 \title{\huge WMMSE-Based Joint Transceiver Design for Multi-RIS Assisted Cell-free Networks Using Hybrid CSI}

\author{Xuesong Pan, Zhong Zheng,~\IEEEmembership{Member, IEEE}, Xueqing Huang,~\IEEEmembership{Member, IEEE}, Zesong Fei,~\IEEEmembership{Senior Member, IEEE}
\thanks{X. Pan, Z. Zheng and Z. Fei are with the School of Information and Electronics, Beijing Institute of Technology, Beijing 100081, China (e-mail: \{xs.pan, zhong.zheng, feizesong\}@bit.edu.cn).}
\thanks{X. Huang is with the Department of Computer Science, New York Institute of Technology, NY 11568, USA (e-mail: xhuang25@nyit.edu).}}


\maketitle

\begin{abstract}
	In this paper, we consider cell-free communication systems with several access points (APs) serving terrestrial users (UEs) simultaneously. To enhance the uplink multi-user multiple-input multiple-output communications, we adopt a hybrid-CSI-based two-layer distributed multi-user detection scheme comprising the local minimum mean-squared error (MMSE) detection at APs and the one-shot weighted combining at the central processing unit (CPU). Furthermore, to improve the propagation environment, we introduce multiple reconfigurable intelligent surfaces (RISs) to assist the transmissions from UEs to APs. Aiming to maximize the weighted sum rate, we formulate the weighted sum-MMSE (WMMSE) problem, where the UEs’ beamforming matrices, the CPU’s weighted combining matrix, and the RISs’ phase-shifting matrices are alternately optimized. Considering the limited fronthaul capacity constraint in cell-free networks, we resort to the operator-valued free probability theory to derive the asymptotic alternating optimization (AO) algorithm to solve the WMMSE problem, which only depends on long-term channel statistics and thus reduces the interaction overhead. Numerical results demonstrate that the asymptotic AO algorithm can achieve a high communication rate as well as reduce the interaction overhead.
\end{abstract}


\section{Introduction}
Cell-free (CF) networks are expected to exhibit promising prospects and widespread applications in mobile communications~\cite{ngo2017cell}. In CF networks, numerous access points (APs) collaboratively serve multiple user equipments (UEs) without the constraints of traditional cell boundaries. Each AP is connected to a central processing unit (CPU) via a fronthaul link, enabling coordinated operation and resource management~\cite{demir2021foundations}. Due to physical constraints, a CF network may include multiple CPUs, each managing the coordination of a subset of APs~\cite{bjornson2020scalable}. Additionally, these CPUs are connected to the core network via backhaul links to facilitate data exchange with the Internet.

With uniform deployment and cooperation of APs, the CF network brings three major advantages as compared to the traditional cellular network: higher signal-to-noise ratio (SNR) with smaller variations, better ability to manage interference, and coherent transmission capabilities that further enhance SNR~\cite{demir2021foundations}. It is worth mentioning that the fronthaul links play a crucial role in performance enhancement, as it is responsible for the interaction of channel state information (CSI) and signals between APs and CPU. However, as CF network scales up, the frequent interaction of the vast instantaneous CSI data between APs and CPU leads to inevitable communications overhead and unpredictable delay, which prevents such architecture from realistic deployment.

To solve this issue, one promising approach is to replace the instantaneous CSI with statistical CSI, where the latter changes about tens of times slower~\cite{huang2011mimo}, thus significantly reducing the interaction overhead between APs and CPU. For instance, a large-scale fading decoding (LSFD) scheme was proposed in~\cite{adhikary2017uplink}, where the joint detection is performed to overcome the inter-cell interference and pilot contamination based on the large-scale fading coefficients. Specifically, in the LSFD scheme, APs first detect the received signals using the local CSI and then forward the detected signals to the CPU. After receiving each AP's locally detected data, CPU will perform joint detection trough linear processing based on the statistical CSI~\cite{bjornson2019making}. The using of local instantaneous CSI at APs and statistical CSI at CPU can not only leverage the macro-diversity gain, but also reduce the interaction overhead.

In the majority of existing works on CF networks, the Rayleigh channel model is generally adopted, which ignores the Line-of-Sight (LoS) propagation between APs and UEs~\cite{bjornson2019making}. This is a reasonable assumption in dense urban areas, where the direct links between APs and UEs may be blocked by obstacles or the antennas of APs are surrounded by many scatters~\cite{bjornson2017massive}. However, the LoS links may exist with high probability in the sub-6 GHz band~\cite{zhi2021statistical} and the LoS propagation becomes the dominant component of channels given the short distances between APs and UEs in CF networks~\cite{wang2024optimal}. Therefore, it's meaningful to consider the Rician channel model to capture the practical channel propagation. 

Due to the absorption by building materials, the multiple-input multiple-output (MIMO) channels typically have only a few LoS propagation paths and limited Non-LoS (NLoS) paths, which causes rank deficiency in the MIMO channel matrix that significantly degrades the channel capacity~\cite{pala2022joint}. To enhance the communication performance, we introduce the multiple reconfigurable intelligent surfaces (RISs) deployment in CF networks. In general, an RIS is a planar surface consisting of a large number of adaptively configurable passive reflecting elements~\cite{Liu2022Dynamic}. Through the nearly passive and tunable signal transformations, i.e., passive beamforming, RISs can realize programmable and reconfigurable wireless propagation environments~\cite{Wang2023Joint}.

The deployment of RISs exhibits benefits in both the NLoS-dominated and the LoS-dominated environments. For the Non-LoS propagation dominated environment, where the LoS connections are blocked, the deployment of RISs can reflect signals and extend the coverage of APs~\cite{wang2022transmit}. For the LoS propagation dominated environment, while other scattered/reflected components are severally attenuated by the building or the terrain, the deployment of RISs can increase the number of independent propagation paths, thereby improving the richness of the end-to-end channel and enhancing the channel capacity~\cite{pala2022joint}.

\subsection{Related Works}
For the NLoS-dominated environment, the RISs are deployed to extend the coverage of APs. The Rayleigh channel model was adopted in \cite{shi2022spatially}, where the uplink throughput of RIS-assisted CF networks was studied. In~\cite{jin2022ris}, the minimum achievable rate was maximized by optimizing the beamforming at APs and phase shifts at RISs, where direct links between APs and UEs are totally blocked by obstacles and RISs are deployed to create cascaded propagation paths and extend the coverage of APs. 
For the LoS-dominated environment, the RISs are deployed to improve the wireless propagation environment, since the few dominating propagation paths will cause rank deficiency in the channel matrix and significantly degrade the MIMO channel capacity~\cite{shin2003capacity}. In~\cite{ma2023cooperative}, the passive and active beamforming at APs and RISs were designed to maximize the weighted sum rate, where the Rician model is adopted in each link between APs and UEs. In~\cite{nguyen2022spectral}, the authors derived the closed-form expressions for downlink and uplink spectral efficiency (SE) of the RIS-assisted CF networks and concluded that the introduction of RISs can enhance the system performance, where the correlated Rician model is adopted. The phase shifts of RISs and beamforming of APs were optimized to enhance the uplink multiple-user MIMO (MU-MIMO) capacity in~\cite{wang2024towards}, where the RISs are deployed to improve the rank of channels between APs and UEs. Therefore, the deployment of RISs can achieve significant performance gains in both scenarios of CF networks. However, the instantaneous CSI is utilized in the design of phase shifts, which will increase the computational complexity and interaction overhead due to the channel estimation and frequent update of phase shifts~\cite{kammoun2020asymptotic}.

On the other hand, considering the limited fronthaul capacity, the LSFD scheme is investigated in CF networks to reduce the interaction overhead~\cite{nayebi2016performance, zhang2021local,mai2019uplink, mai2020downlink, wang2022uplink, wang2023uplink}. The authors in \cite{nayebi2016performance} investigated the uplink performance of CF networks, where the single-antenna APs first detect UEs' signals using the match filter (MF) and the CPU conducts LSFD to obtain the final detection. To further reduce interference, the local zero-forcing (ZF) detector with the LSFD receiver scheme in the Rayleigh fading channels was studied in~\cite{zhang2021local}. In~\cite{bjornson2019making}, the authors considered the local minimum mean-squared error (MMSE) detector in APs and the LSFD receiver at the CPU, which tradeoffs between signal enhancement and interference suppression. Note that the above works all assume that the UEs are equipped with single antenna. Since the contemporary UEs have been equipped with multiple antennas for additional spatial degree of freedom exploitation~\cite{pan2024uplink}. The authors investigated the uplink and downlink performance of CF networks with multi-antenna UEs in the uncorrelated Rayleigh fading channels in~\cite{mai2019uplink, mai2020downlink}. Inspired by~\cite{bjornson2019making}, the authors studied the uplink performance of local MMSE detector with LSFD receiver scheme in CF networks with multi-antenna UEs, and then investigated the precoding design of multi-antenna UEs in the correlated Rayleigh fading channels in~\cite{wang2022uplink,wang2023uplink}. However, the existing works widely adopt the NLoS radio propagation between APs and UEs in CF networks. The inclusion of LoS component will introduce significant challenges and different difficulties in the optimization of system performance~\cite{liu2023cell}.

\subsection{Contributions}
To address these drawbacks in existing works, we investigate the uplink MU-MIMO detection in multi-RIS-assisted CF networks, where a two-layer signal detection scheme is adopted to jointly recover the multi-antenna UEs’ signals based on hybrid CSI, i.e. the instantaneous CSI and statistical CSI. In particular, the APs first detect the UEs’ signals by using the MMSE detector based on the local instantaneous CSI. Then, the CPU collects the detected signals from all APs and performs weighted combining based on the long-term statistical CSI. The weighted combining matrix (WCM) of the CPU, the transmit precoding matrices (TPMs) of the UEs, and the phase-shifting matrices (PSMs) of the RIS panels are jointly optimized to maximize the weighted sum rate. In the proposed system architecture, the main contributions of this paper are summarized as follows:
\begin{itemize}
\item The general Rician MIMO model with Weichselberger’s correlation structure is adopted for the channel between any two nodes. The channel model fits a wide range of realistic MIMO channels that have LoS propagation components and arbitrary scattering distribution, i.e., a general framework that is applicable to both NLoS-dominated and LoS-dominated environments. Therefore, this channel model can better capture the characteristics of the propagation links in CF networks than the ones in existing related works.
\item Following the weighted sum-MMSE (WMMSE) framework in~\cite{shi2011iteratively}, the formulated original non-convex weighted sum-rate maximization is converted into a WMMSE problem, which can be further decomposed into a sequence of simpler sub-problems that alternately optimize WCM, TPMs and PSMs. Therefore, these optimizing variables are decoupled and can be solved by the MMSE algorithm, the Lagrange multipliers method and the Riemannian manifold based gradient descent algorithm, respectively. Numerical results demonstrate the fast convergence and the performance gain of the overall iterative WMMSE algorithm.
\item Taking the limited fronthual capacity into account, we resort to the operator-valued free probability theory to further reformulate the decomposed sub-problems and then propose the asymptotic alternating optimization (AO) algorithm to obtain the optimal transceivers and phase shifts, which only depend on the statistical CSI and avoid time-consuming Monte Carlo simulations. The proposed algorithm can effectively reduce the interaction overhead and lay the foundation for the practical deployment of multi-RIS-assisted CF networks.

\end{itemize}

The rest of this article is organized as follows. The system model and the two-layer signal detection architecture are introduced in Section~\ref{sec_sysmodel}. The AO algorithm and asymptotic AO algorithm for the WCM, TPMs and PSMs is respectively proposed in Section~\ref{sec_AO} and Section~\ref{sec_asymAO}. Numerical results are given in Section~\ref{sec_Result}. Section~\ref{sec_conclusion} concludes the main findings of this article.

\emph{Notations.} Throughout the paper, vectors are represented by lower-case bold-face letters, and matrices are represented by upper-case bold-face letters. The superscripts $(\cdot)^*$, $(\cdot)^\text{T}$ and $(\cdot)^\dag$ denote the conjugate, transpose and conjugate transpose operations, respectively. We use $\mb{0}_n$ and $\mb{I}_n$ to represent an $n\times n$ all-zero matrix and an $n\times n$ identity matrix, respectively. The notation ${\rm{blkdiag}}(\mb{A}_1,\dots,\mb{A}_n)$ denotes the diagonal block matrix consisting of square matrices $\mb{A}_1,\dots,\mb{A}_n$. The notation $[\mb{A}]_{mn}$ denotes the $(m,n)$-th entry of matrix $\mb A$. The notations $\otimes$ and $\odot$ are denoted as Kronecker product and Hadamard product, respectively. For symmetric matrices $\mb{A}$ and $\mb{B}$, $\mb{A}\succeq\mb{B}$ signifies that $\mb{A}-\mb{B}$ is positive semidefinite. The operator $\to_{a.s.}$ denotes \textit{almost surely converge to}.

\section{System Model and Uplink Signal Detection Architecture} \label{sec_sysmodel}
We consider a multi-RIS-assisted CF network comprising $L$ APs, $N$ UEs and $K$ RISs, as illustrated in Fig.~\ref{figModel}. Without loss of generality, each AP connects to one CPU via fronthaul, which is used to share CSI and signals with the CPU and facilitate the phase-synchronization with the other APs~\cite{demir2021foundations}. The RISs are controlled by the CPU via a controller, which is used to exchange information and achieve phase synchronization with APs and the CPU. The signal model and channel model are given as follows.
\subsection{Signal Model}
 In this system, each AP and UE are equipped with $R$ receiving antennas and $T$ transmitting antennas, respectively. Each RIS panel is equipped with $L_k,1\leq k\leq K$ passive reflecting elements. We define the total numbers of transmitting antennas and passive reflecting elements as $T_t=N\times T$ and $L_R=\sum_{k=1}^{K}L_k$, respectively. In addition, we define $L_{AR}=R+L_R$.

Let $\mb{F}_{0,nl}\in \mbb{C}^{R\times T}$, $\mb{F}_{nk}\in \mbb{C}^{L_k\times T}$ and $\mb{G}_{kl}\in\mbb{C}^{R\times L_k}$ denote the direct channel from the $n$-th UE to the $l$-th AP, the $n$-th UE to the $k$-th RIS and the $k$-th RIS to the $l$-th AP, respectively. Denote $\mb{x}_n\in\mbb{C}^{T}$ as the transmitting signal of the $n$-th UE, where $\mb{x}_n\in\mc{CN}(\mb{0},\mb{I}_{T})$, and $\mb{n}_l\in \mbb{C}^{R}$ as the additive white Gaussian noise at the $l$-th AP, where $\mb{n}_l\in \mc{CN}(\mb{0},\sigma^2\mb{I}_{R})$ with $\sigma^2$ being the noise power. The received signal $\mb{y}_l\in\mbb{C}^{R}$ at the $l$-th AP is expressed as
\begin{equation} \label{equ_Signalmodel}
	\begin{aligned}
		\mb{y}_l = \sum_{n=1}^{N}\left(\mb{F}_{0,nl} + \sum_{k=1}^{K}\mb{G}_{kl}\mb{\Theta}_{k}\mb{F}_{nk}\right)\mb{W}_{n}\mb{x}_{n}+\mb{n}_{l},
	\end{aligned}
\end{equation}
where $\mb{\Theta}_{k}=\operatorname{diag}(e^{j\theta_{k,1}},\cdots,e^{j\theta_{k,L_{k}}})$ denotes the PSM of the $k$-th RIS, where $0\leq \theta_{k,l_k}\leq 2\pi$ for $l_k = 1,\dots,L_k$, and $\mb{W}_{n}\in\mbb{C}^{T\times T}$ denotes the TPM of the $n$-th UE, which satisfies the power constraint as $\mathrm{Tr}(\mb{W}_{n}\mb{W}_{n}^\dag)\leq p_n$ with $p_n$ being the maximum transmitting power for the $n$-th UE.

Let $\mb{F}_l=[\mb{F}_{0l}^\dag,\mb{F}_1^\dag,\dots,\mb{F}_K^\dag]^\dag\in \mbb{C}^{L_{AR}\times T_t}$ with $\mb{F}_{0l}=[\mb{F}_{0,1l},\dots,\mb{F}_{0,Nl}]\in\mbb{C}^{R\times T_t}$ and $\mb{F}_k=[\mb{F}_{1k},\dots,\mb{F}_{Nk}]\in\mbb{C}^{L_k\times T_t}$. Also, let $\mb{G}_l=[\mb{I}_{R},\mb{G}_{1l},\dots,\mb{G}_{Kl}]\in\mbb{C}^{R\times L_{AR}}$ and $\mb{\wh{\Theta}} = \mathrm{blkdiag}\left(\mb{I}_{R}, \mb{\Theta}_{1},\dots,\mb{\Theta}_K\right)\in\mbb{C}^{L_{AR}\times L_{AR}}$. We define $\mb{H}_l=\mb{G}_l\mb{\wh{\Theta}}\mb{F}_l\in\mbb{C}^{R\times T_t}$ and $\mb{H}_{nl}=\mb{G}_l\mb{\wh{\Theta}}\mb{F}_{nl}\in\mbb{C}^{R\times T}$ with $\mb{F}_{nl}=[\mb{F}_{0,nl}^\dagger,\mb{F}_{n1}^\dagger,\dots,\mb{F}_{nK}^\dagger]^\dagger\in\mbb{C}^{L_{AR}\times T}$. Then, the received signal (\ref{equ_Signalmodel}) can be rewritten as 
\begin{equation}
	\begin{aligned}
		\mb{y}_l = \mb{H}_l\mb{\wh {W}}\mb{x}+\mb{n}_l= \sum_{n=1}^{N}\mb{H}_{nl}\mb{W}_{n}\mb{x}_{n}+\mb{n}_{l},
	\end{aligned}
\end{equation}
where $\mb{\wh {W}}=\operatorname{blkdiag}\left(\mb{W}_{1},\dots,\mb{W}_{N}\right)\in\mbb{C}^{T_t\times T_t}$ and $\mb{x}=[\mb{x}_{1}^\dag,\dots,\mb{x}_N^\dag]^\dag\in \mbb{C}^{T_t}$. 
\begin{figure}[tb] 
	\centerline{\includegraphics[width=0.9\columnwidth]{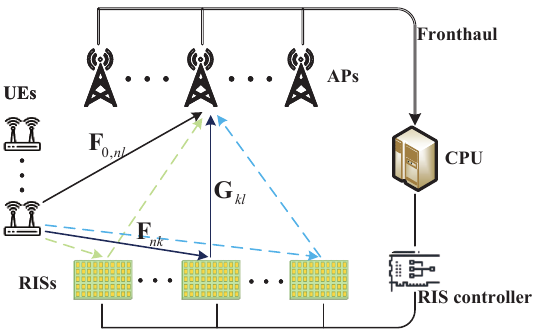}}
	\caption{RIS-assisted CF networks with capacity-limited fronthaul links.}
	\label{figModel}
\end{figure}

\subsection{Channel Model}
We adopt the non-central Weichselberger's MIMO channel model~\cite{weichselberger2006stochastic} for each component channel in (\ref{equ_Signalmodel}), such that
\begin{align}
	\mb{F}_{0,nl} & = \overline{\mb{F}}_{0,nl} + \widetilde{\mb{F}}_{0,nl} \notag \\
	&= \overline{\mb{F}}_{0,nl} + \mb{T}_{0,nl} (\mb{M}_{0,nl} \odot \mb{X}_{0,nl}) \mb{P}_{0,nl}^{\dagger}, \label{equ_F0}\\
	\mb{F}_{nk} & = \overline{\mb{F}}_{nk} + \widetilde{\mb{F}}_{nk} = \overline{\mb{F}}_{nk} + \mb{T}_{nk} (\mb{M}_{nk} \odot \mb{X}_{nk}) \mb{P}_{nk}^{\dagger},\label{equ_F}\\
	\mb{G}_{kl} &= \overline{\mb{G}}_{kl} + \widetilde{\mb{G}}_{kl} = \overline{\mb{G}}_{kl} + \mb{R}_{kl} (\mb{N}_{kl}\odot\mb{Y}_{kl}) \mb{C}_{kl}^{\dagger}, \label{equ_G}
\end{align}
where the deterministic matrices $\overline{\mb{F}}_{0,nl}$, $\overline{\mb{F}}_{nk}$ and $\overline{\mb{G}}_{kl}$ denote the LoS component of $\mb{F}_{0,nl}$, $\mb{F}_{nk}$ and $\mb{G}_{kl}$, respectively. The random scattering components are captured by $\widetilde{\mb{F}}_{0,nl}$, $\widetilde{\mb{F}}_{nk}$ and $\widetilde{\mb{G}}_{kl}$, where $\mb{T}_{0,nl}$, $\mb{P}_{0,nl}$, $\mb{T}_{nk}$, $\mb{P}_{nk}$, $\mb{R}_{kl}$ and $\mb{C}_{kl}$ are deterministic spatial correlation matrices and the deterministic matrices $\mb{M}_{0,nl}$, $\mb{M}_{nk}$ and $\mb{N}_{kl}$ with nonnegative entries represent the variance profiles of $\widetilde{\mb{F}}_{0,nl}$, $\widetilde{\mb{F}}_{nk}$ and $\widetilde{\mb{G}}_{kl}$, respectively. The random matrices $\mb{X}_{0,nl}\in\mbb{C}^{R\times T}$, $\mb{X}_{nk}\in\mbb{C}^{L_k\times T}$ and $\mb{Y}_{kl}\in\mbb{C}^{R\times L_k}$ are complex Gaussian distributed with zero-mean and independent entries, i.e., $[\mb{X}_{0,nl}]_{ij}\in\mc{CN}(0,1/T)$, $[\mb{X}_{nk}]_{ij}\in\mc{CN}(0,1/T)$ and $[\mb{Y}_{kl}]_{ij}\in\mc{CN}(0,1/L_k)$. In addition, we assume that the frequency-flat fading channels of different links are independent, since the channels $\{\mb{F}_{0,nl}\}_{1\leq n \leq N, 1\leq l\leq L}$, $\{\mb{F}_{nk}\}_{1\leq n \leq N, 1\le k\le K}$ and $\{\mb{G}_{kl}\}_{1\le k\le K, 1\leq l\leq L}$ are spatially separated.

\subsection{Two-layer Signal Detection Scheme} \label{subsec_Oneshot}
In this subsection, we propose the distributed MMSE detection with one-shot combining scheme, including the local MMSE detection in APs and one-shot weighted combining at the CPU. Specifically, in the first layer, each AP individually recovers UEs' signals by using the MMSE detector based on the local instantaneous CSI. By minimizing the MSE between the locally recovered signal $\mb{\wh{x}}_{nl}=\mb{U}_{nl}^\dag\mb{y}_l$ and the original signal $\mb{{x}}_{nl}$, i.e., $\mbb{E}\lVert\mb{{x}}_{n}-\mb{\wh{x}}_{nl}\rVert^2$, the MMSE detector $\mb{U}_{nl}\in\mbb{C}^{R\times T}$ can be expressed as 
\begin{equation} \label{equ_U}
	\mb{U}_{nl}=(\mb{H}_l\mb{\wh {W}}\mb{\wh {W}}^\dag\mb{H}_l^\dag+\sigma^2\mb{I}_{R})^{-1}\mb{H}_{nl}\mb{W}_{n}.
\end{equation}

The MMSE detector only depends on the local instantaneous CSI and does not require the information interaction between APs. To exploit the macro-diversity gain and achieve joint detection, after the local detection, each AP transmits the locally detected signals to the CPU for further refinement.

In the second layer, after collecting the locally recovered signals from every AP, the CPU then combines the signals as
\begin{equation}
	\mb{\wh{x}}_{n} = \sum_{l=1}^{L}\mb{A}_{nl}^\dag\mb{\wh{x}}_{nl},
\end{equation}
where $\mb{\wh{x}}_{n}\in\mbb{C}^{T}$ is the final refined estimate of the $n$-th UE's signal and $\mb{A}_{nl}\in\mbb{C}^{T\times T}$ is the one-shot WCM assigned to the detected signals of the $n$-th UE at the $l$-th AP, which is required to be optimized to maximize the sum rate. Defining $\mb{A}_n=[\mb{A}_{n1}^\dag,\dots,\mb{A}_{nL}^\dag]^\dag\in\mbb{C}^{LT\times T}$ and $\mb{\Phi}_{nm}=[\mb{U}_{n1}^\dag\mb{H}_{m1};\dots;\mb{U}_{nL}^\dag\mb{H}_{mL}]\in\mbb{C}^{LT\times T}$, then the combined signals of the $n$-th UE can be rewritten as
\begin{equation} \label{equ_FinalSignal}
	\mb{\wh{x}}_{n} = \sum_{m=1}^{N}\mb{A}_n^\dag\mb{\Phi}_{nm}\mb{W}_m\mb{x}_m+\sum_{l=1}^{L}\mb{A}_{nl}^\dag\mb{U}_{nl}^\dag\mb{n}_l.
\end{equation}

Based on the final refined signal in (\ref{equ_FinalSignal}), the achievable rate of the $n$-th UE is given by
\begin{equation}
	R_n=\operatorname{log}\operatorname{det}\left(\mb{I}_{T}+\mb{D}_n^\dag\mb{\Sigma}_n^{-1}\mb{D}_n\right),
\end{equation}
where $\mb{D}_n=\mb{A}_n^\dag\mbb{E}(\mb{\Phi}_{nn})\mb{W}_n$ and $\mb{\Sigma}_n=\sigma^2\mb{A}_n^\dag\mbb{E}(\mb{S}_n)\mb{A}_n+\sum_{m=1}^{N}\mb{A}_n^\dag\mbb{E}(\mb{\Phi}_{nm}\mb{W}_m\mb{W}_m^\dagger\mb{\Phi}_{nm}^\dagger)\mb{A}_n-\mb{D}_n\mb{D}_n^\dagger$ with $\mb{S}_n=\operatorname{blkdiag}(\mb{U}_{n1}^\dag\mb{U}_{n1},\dots,\mb{U}_{nL}^\dag\mb{U}_{nL})$.
The achievable rate is obtained via the lower bound of the mutual information, where the expectation therein is taken with respect to the instantaneous CSI~\cite{wang2022uplink}.

\section{Joint Optimization for Weighted Combining Matrix, Transmit Precoding Matrices, and Phase-Shifting Matrices} \label{sec_AO}
 Under the two-layer signal detection architecture, the UEs' signal can be jointly detected to achieve high communication rate. In addition to the design of the receiver, the communication quality can be also effectively enhanced by optimizing the TPMs at UEs, which can exploit the spatial diversity, and the PSMs at RISs, which can improve the wireless propagation environments. In this section, we investigate the weighted sum-rate maximization problem by jointly optimizing the WCM at the CPU, TPMs at UEs, and PSMs at RISs.
 
We aim to maximize the weighted sum rate as follows
\begin{align}
	\max_{\bm{\mc A},\bm{\mc W},\bm{\varTheta}}\quad  
	&\sum_{n=1}^{N}\mu_n R_n(\bm{\mc A},\bm{\mc W},\bm{\varTheta})  \label{equ_WSR}\\
	~{\rm {s.t.}}~\quad   
	&\mathrm{Tr}(\mb{W}_{n}\mb{W}_{n}^\dag)\leq p_n, 1\leq n\leq N, \tag{\ref{equ_WSR}a} \label{equ_WSRa}\\
	\qquad \quad    	    &\mb{\Theta}_{k}=\operatorname{diag}(e^{j\theta_{k,1}},\cdots,e^{j\theta_{k,L_{k}}}),1\leq k\leq K,	\tag{\ref{equ_WSR}b}	\label{equ_WSRb}
\end{align}
where $\mu_n$ represents the priority of the $n$-th UE. 
The variables $\bm{\mc A}=\{\mb{A}_1,\dots,\mb{A}_N\}$, $\bm{\mc W}=\{\mb{W}_1,\dots,\mb{W}_N\}$ and $\bm{\varTheta}=\{\mb{\Theta}_{1},\dots,\mb{\Theta}_{K}\}$. The constraints (\ref{equ_WSRa}) and (\ref{equ_WSRb}) represent the transmit power constraint and phase shift range constraint, respectively.

Following~\cite{shi2011iteratively}, the weighted sum-rate maximization problem  given in~(\ref{equ_WSR}) is equivalent to the matrix-weighted sum-MSE minimization problem, which is given by
\begin{align} 
	\min_{\bm{\mc A},\bm{\mc W},\bm{\varTheta},\bm{\mc{V}}}\quad  
	&\sum_{n=1}^{N}\mu_n \Big(\operatorname{Tr}(\mb{V}_n\mb{E}_n(\bm{\mc A},\bm{\mc W},\bm{\varTheta}))-\operatorname{log}\operatorname{det}(\mb{V}_n)\Big)  \label{equ_WMMSE}\\
	~{\rm {s.t.}}~\quad & (\ref{equ_WSR}\text{a}),(\ref{equ_WSR}\text{b}), \notag
\end{align}
where $\bm{\mc{V}}=\{\mb{V}_{1},\dots,\mb{V}_{N}\}$ with $\mb{V}_n\succeq \mb{0}_T $ being the auxiliary weight matrix variable for the $n$-th UE. The MSE matrix $\mb{E}_n$ can be written as 
\begin{equation} \label{equ_E}
	\begin{aligned}
		\mb{E}_n &= \mbb{E}\left\{(\mb{x}_n-\mb{\wh{x}}_{n})(\mb{x}_n-\mb{\wh{x}}_{n})^\dagger\right\}\\
		&=\mb{I}_{T}-\mb{W}_n^\dag\mbb{E}(\mb{\Phi}_{nn})^\dagger\mb{A}_n-\mb{A}_n^\dagger\mbb{E}(\mb{\Phi}_{nn})\mb{W}_n\\
		&\quad+\mb{A}_n^\dagger\left(\sum_{m=1}^{N}\mbb{E}(\mb{\Phi}_{nm}\mb{W}_m\mb{W}_m^\dagger\mb{\Phi}_{nm}^\dagger)+\sigma^2\mbb{E}(\mb{S}_n)\right)\mb{A}_n.
	\end{aligned}
\end{equation}

The problem (\ref{equ_WMMSE}) is complicated due to the coupled optimization variables and the non-convex unit-modulus constraints of PSMs, i.e., $\mb{\Theta}$ in (\ref{equ_WSR}\text{b}). Therefore, we propose the AO algorithm to solve the problem (\ref{equ_WMMSE}), which decomposes the original optimization problem into three sub-problems, i.e., the WCM sub-problem, the TPM sub-problem, and the PSM sub-problem, and then optimize the corresponding variables in an alternative manner.


\subsection{The WCM sub-problem}
Fixing the TPMs $\mb{W}$ and the PSMs $\mb{\Theta}$, the optimal one-shot WCM $\mb{A}$ can be obtained by minimizing the sum-MSE problem as follows
\begin{align} \label{equ_MSE}
	\min_{\bm{\mc A}}\quad  
	\sum_{n=1}^{N} \operatorname{Tr}(\mb{E}_n(\bm{\mc A})).
\end{align}

By checking the first-order optimality condition of (\ref{equ_MSE}) with respect to $\mb{A}_n$, the optimal WCM is given by
\begin{equation} \label{equ_optA}
	\mb{A}^{\rm opt}_n = \mbb{E}\left\{\sum_{m=1}^{N}\mb{\Phi}_{nm}\mb{W}_m\mb{W}_m^\dagger\mb{\Phi}_{nm}^\dagger+\sigma^2\mb{S}_n\right\}^{-1}\mbb{E}(\mb{\Phi}_{nn})\mb{W}_n.
\end{equation}

\subsection{The update of weight matrix variables}
Since the objective function (\ref{equ_WMMSE}) is convex with respect to the auxiliary weight matrix variable $\mb{V}_n$, through the first-order optimality
condition for $\mb{V}_n$ and fixing the other variables, we have
\begin{equation}
	\mb{V}^{\rm opt}_n = (\mb{E}_n^{\rm opt})^{-1},
\end{equation}
where $\mb{E}_n^{\rm opt}$ can be obtained by substituting $\mb{A}^{\rm opt}$ into (\ref{equ_E}).

\subsection{The TPM sub-problem}
By fixing $\mb A$, ${\mb V}$ and $\mb{\Theta}$, the optimization problem (\ref{equ_WMMSE}) with respect to the TPMs $\mb{W}$ is reduced to
\begin{align}
	\min_{\bm{\mc W}}\quad  
	&\sum_{n=1}^{N}\mu_n \operatorname{Tr}(\mb{V}_n\mb{E}_n(\bm{\mc W})) \label{equ_WMMSE_W}\\
	~{\rm {s.t.}}~\quad   
	& (\ref{equ_WSR}\text{a}). \notag
\end{align}
Note that the above problem is a convex quadratic optimization problem, which can be solved by using standard convex optimization algorithms. By applying the Lagrange multipliers method and attaching a Lagrange multiplier $\lambda_n$ to the power budget constraint of the $n$-th UE, the Lagrange function can be expressed as
\begin{equation}\label{equ_L}
	\begin{aligned}
		L(\mb{W}_n,\lambda_n) = 
		\sum_{n=1}^{N}&\mu_n \operatorname{Tr}(\mb{V}_n\mb{E}_n(\bm{\mc W}))\\
		&+\lambda_n(\mathrm{Tr}(\mb{W}_{n}{\mb{W}_{n}}^\dag)-p_n).
	\end{aligned}
\end{equation}

To obtain the optimal TPM of the $n$-th UE, we substitute (\ref{equ_E}) into (\ref{equ_L}) and remove some irrelevant items. Then, the Lagrange function can be reduced to (\ref{equ_eqL}), shown at the bottom of this page. The first-order optimal condition of $L(\mb{W}_n,\lambda_n)$ leads to the optimal TPM of the $k$-th UE shown as follows
\setcounter{equation}{18}
\begin{equation} \label{equ_optW}
	\begin{aligned}
		\mb{W}^{\rm opt}_n =
		\mu_n& \left(\sum_{m=1}^{N}\mu_m\mathbb{E}\left(\mb{\Phi}_{mn}^\dagger\mb{A}_m\mb{V}_m\mb{A}_m^\dagger\mb{\Phi}_{mn}\right)+\lambda_n\mb{I}_T\right)^{-1} \\
		&\qquad\times \mathbb{E}(\mb{\Phi}_{nn}^\dagger)\mb{A}_n\mb{V}_n,
	\end{aligned}
\end{equation}
where the Lagrange multiplier $\lambda_k\geq 0$ can be obtained by a one-dimensional (1-D) bisection algorithm, according to the Karush-Kuhn-Tucker (KKT) conditions~\cite{shi2011iteratively}.
 
\begin{figure*}[b] 
	\hrulefill  
	\begin{align} 
		\setcounter{equation}{17}
		&L(\mb{W}_n,\lambda_n) = \sum_{m=1}^{N}\mu_m\Big\{\operatorname{Tr}(\mb{V}_m)-\underbrace{\mathbb{E}\left\{\operatorname{Tr}\left(\mb{V}_m(\mb{W}_m)^\dagger\mb{\Phi}_{mm}^\dagger\mb{A}_m\right)\right\}}_{(\ref{equ_eqL}a)}-\underbrace{\mathbb{E}\left\{\operatorname{Tr}\left(\mb{V}_m\mb{A}_m^\dagger\mb{\Phi}_{mm}\mb{W}_m\right)\right\}}_{(\ref{equ_eqL}b)} \notag \\	&\quad+\underbrace{\mathbb{E}\left\{\operatorname{Tr}\left((\mb{W}_n)^\dagger\mb{\Phi}_{mn}^\dagger\mb{A}_m\mb{V}_m\mb{A}_m^\dagger\mb{\Phi}_{mn}\mb{W}_n\right)\right\}}_{(\ref{equ_eqL}c)}\Big\}+\lambda_n\left(\mathrm{Tr}(\mb{W}_{n}({\mb{W}_{n}})^\dag)-p_n\right). \label{equ_eqL}\\
		\setcounter{equation}{20}
		&f_n(\mb{\bm{\varTheta}}) =   \operatorname{Tr}\Big(\mb{V}_n\mb{A}_n^\dagger\big(\sum_{m=1}^{N}\mbb{E}(\mb{\Phi}_{nm}\mb{W}_m\mb{W}_m^\dagger\mb{\Phi}_{nm}^\dagger)+\sigma^2\mbb{E}(\mb{S}_n)\big)\mb{A}_n\Big)-\operatorname{Tr}(\mb{V}_n\mb{W}_n^\dag\mbb{E}(\mb{\Phi}_{nn})^\dagger\mb{A}_n)
		-\operatorname{Tr}(\mb{V}_n\mb{A}_n^\dagger\mbb{E}(\mb{\Phi}_{nn})\mb{W}_n). \label{equ_f} 
	\end{align}
\end{figure*}

\subsection{The PSM sub-problem} \label{sec_PSM}
By fixing the updated variables $\mb A$, $\mb V$ and $\mb W$, the optimization problem (\ref{equ_WMMSE}) with respect to the PSMs $\mb{\Theta}$ is
\setcounter{equation}{19}
\begin{align} 
	\min_{\bm{\varTheta}}\quad  
	&\sum_{n=1}^{N}\mu_n f_n(\bm{\varTheta})  \label{equ_optif}\\
	~{\rm {s.t.}}~\quad     	    &  (\ref{equ_WSR}\text{b}), \notag
\end{align}
where the objective function $f_n(\bm{\varTheta})$ can be expressed as (\ref{equ_f}), shown at the bottom of this page. 

To solve this non-convex problem, we utilize the Riemannian manifold based gradient descent algorithm to obtain the optimal PSMs~\cite{wang2021joint}.
Defining $\bm{\theta}=\operatorname{blkdiag}(\mb{\Theta}_1,\dots,\mb{\Theta}_K)\mb{1}_{L_R}$, where $\mb{1}_{L_R}\in\mbb{C}^{L_R}$ denotes an all-one vector. The unit-modulus constraint (\ref{equ_WSR}\text{b}) defines an oblique manifold, which can be characterized by
\setcounter{equation}{21}
\begin{equation}
	\mc{O}=\left\{\bm{\theta}\in\mbb{C}^{L_R}\vert\lvert\bm{\theta}_k\rvert=1,1\leq k \leq L_R\right\},
\end{equation}
where the notation $\boldsymbol{\theta}_k$ denotes the $k$-th element of vector $ \boldsymbol{\theta} $ and the operator $|\boldsymbol{\theta}_k|$ means the modulus of $ \boldsymbol{\theta}_k$.

For a given point $\bm{\theta}$ on the manifold $\mc{O}$, the tangent space $\mc{T}$ can be defined as
\begin{equation}
	\mc{T}=\left\{\mb{t}\in\mbb{C}^{L_R}\vert[\mb{t}\bm{\theta}^\dagger]_{kk},\bm{\theta}\in\mc{O},k=1,\dots,L_R\right\},
\end{equation}
where $\mb{t}$ is the tangent vector at the given point $\bm{\theta}$. Similar to the Euclidean space, the direction of the fastest increase of the function ${f}\left(\bm{\varTheta}\right)=\sum_{n=1}^{N}\mu_n {f}_n(\bm{\varTheta})$ among all the tangent vectors is defined as the Riemannian gradient, given by~\cite{wang2021joint}
\begin{equation}
	\begin{aligned}
		\operatorname{grad}_{\bm{\theta}}{f}=\nabla {f}(\bm{\theta}) -\operatorname{Re}\left\{\nabla {f}(\bm{\theta})\odot\bm{\theta}^\star\right\}\odot\bm{\theta},
	\end{aligned}
\end{equation}
where $\nabla {f}(\bm{\theta})=2\frac{d{f}}{d\bm{\theta}^\star}$ denotes the Euclidean gradient of ${f}$ with respect to $\bm{\theta}$ and the derivative $\frac{d{f}}{d\bm{\theta}^\star}$ is given by
\begin{equation}
	\begin{aligned}
		\frac{d{f}}{d\bm{\theta}^\star} 
		= \left[\frac{d{f}}{d\bm{\theta}_{1}^\star},\dots,\frac{d{f}}{d\bm{\theta}_{L_R}^\star}\right]^{\rm T}
		=\left[\frac{d{f}}{d\mb{\Theta}_{1,1}^\star},\dots,\frac{d{f}}{d\mb{\Theta}_{K,L_K}^\star}\right]^{\rm T},
	\end{aligned}
\end{equation}
where $\mb{\Theta}_{k,l_k}$ denotes the $l_k$-th phase-shift of the $k$-th RIS.
The expression of the derivative $\frac{d{f}}{d\mb{\Theta}_{k,l_k}^\star}$ can be directly obtained based on the matrix derivation. For lack of space, the expressions are omitted here.

Based on the Riemannian steepest descent algorithm, the descent direction at the $i$-th iteration is chosen as the Riemannian gradient, i.e., $\bm{\delta}^{(i)}=\operatorname{grad}_{\bm{\theta}^{(i)}}{f}$. Thus, in the $(i + 1)$-th iteration, we update $\bm{\theta}_k^{(i+1)}$ by
\begin{equation}  \label{equ_update_theta}
	\bm{\theta}^{(i+1)}_k=\frac{\bm{\theta}^{(i)}_k+\alpha^{(i)}_k\bm{\delta}^{(i)}_k}{\lvert\bm{\theta}^{(i)}_k+\alpha^{(i)}_k\bm{\delta}^{(i)}_k\rvert},
\end{equation}
where $\alpha^{(i)}\in\mbb{C}^{L_R}$ is the step size in the $i$-th iteration and the operation of the denominator aims to guarantee the constant-modulus constraint of $\mb{\Theta}$.

Note that expressions of WCM, TPMs and PSMs are in semi-closed forms, where the channel samples must be exchanged between APs and CPU in order to obtain the statistical expectations in (\ref{equ_optA}), (\ref{equ_optW}) and (\ref{equ_f}). The heavy and inevitable information interaction overhead prevents the joint optimization framework from being realistically deployed. 

\section{Operator-Valued Free Probability Theory for Asymptotic AO Algorithm Design} \label{sec_asymAO} 

In this section, we first reformulate each optimization sub-problem as a function of the Cauchy transform. Leveraging the operator-valued free probability theory, the exact expressions of the Cauchy transform can be derived by using statistical CSI. Consequently, the asymptotic expression of each sub-problem in the AO framework can be obtained and alternately solved based on the statistical CSI at the CPU. Since the statistical CSI changes slower than small-scale fading coefficients and is assumed to be constant over several coherent blocks~\cite{nayebi2016performance}, the statistical CSI only needs to be transmitted once from APs to CPU, which avoids frequent information exchanges between APs and the CPU and greatly reduce the interaction overhead.
To clarify the iteration process, the variables in the $i$-th iteration are denoted as $\mb{A}^{(i)}$, $\mb{W}^{(i)}$, and $\mb{\Theta}^{(i)}$. The optimization in the $(i+1)$-th iteration is given as follows.

\subsection{Optimization of $\mb A^{(i+1)}$ for given $\mb W^{(i)}$ and $\mb{\Theta}^{(i)}$} \label{sec_optiA}
Recalling the sum-MSE problem (\ref{equ_MSE}), the objective function can be denoted as 
\begin{equation} \label{equ_TrE}
	\begin{aligned}
		\operatorname{Tr}(\mb{E}_n) &= T-	\underbrace{\operatorname{Tr}(\mb{W}_n^\dag\mbb{E}(\mb{\Phi}_{nn})^\dagger\mb{A}_n)}_{(\ref{equ_TrE}\mathrm{a})}-	\underbrace{\operatorname{Tr}(\mb{A}_n^\dagger\mbb{E}(\mb{\Phi}_{nn})\mb{W}_n)}_{(\ref{equ_TrE}\mathrm{b})}\\
		+	&\underbrace{\operatorname{Tr}\left(\mb{A}_n^\dagger\mbb{E}\Big\{\sum_{m=1}^{N}\mb{\Phi}_{nm}\mb{W}_m\mb{W}_m^\dagger\mb{\Phi}_{nm}^\dagger+\sigma^2\mb{S}_n\Big\}\mb{A}_n\right)}_{(\ref{equ_TrE}\mathrm{c})},
	\end{aligned}
\end{equation}
 Note that the expectation matrices in (\ref{equ_TrE}) are based on the time-consuming numerical simulations and coupled with optimization variables $\mb{W}$ and $\mb{\Theta}$, which need frequent interaction between APs and CPU. To facilitate calculation and reduce interaction overhead, we resort to the Cauchy transform to obtain the asymptotic expressions of (\ref{equ_TrE}), given as follows.
 
Substituting (\ref{equ_U}) into (\ref{equ_TrE}a) and invoking the matrix inversion lemma~\cite{couillet2011random}, the equation (\ref{equ_TrE}a) can be rewritten as 
\begin{equation}
	\begin{aligned}	&\operatorname{Tr}(\mb{W}_n^\dag\mbb{E}(\mb{\Phi}_{nn})^\dagger\mb{A}_n) 
	=\sum_{l=1}^{L}\mbb{E}\left\{\operatorname{Tr}(\mb{W}_n^\dag\mb{H}_{nl}^\dag\mb{U}_{nl}\mb{A}_{nl})\right\}\\
	&=\sum_{l=1}^{L}\mbb{E}\left\{\operatorname{Tr}\left\{\mb{A}_{nl}\left\{\mb{I}_{T_t}-\sigma^2(\sigma^2\mb{I}_{T_t}+\mb{B}_l)^{-1}\right\}_{nn}\right\}\right\},
	\end{aligned}
\end{equation}
where the notation $\{\}_{nn}$ denotes the $n$-th diagonal block matrix of dimension $T\times T$ and the Gram matrix $\mb{B}_l=\mb{\wh {W}}^\dag\mb{H}_l^\dagger\mb{H}_l\mb{\wh {W}}$. Consider the large-dimensional regime\footnote{The proposed algorithm is still applicable the scenarios with a finite number of antennas and passive reflecting elements and the approximate accuracy is always at a high level, which can be illustrated in Fig.~\ref{Fig_Accuracy}.} as
	\begin{equation}
		R\to\infty,T\to\infty,L_R\to \infty, R/T=\epsilon_1, L_R/T=\epsilon_2,
	\end{equation}
	where $\epsilon_1$ and $\epsilon_2$ are constants. According to the operator-valued free probability theory~\cite{lu2016free}, the Cauchy transform of $\mb{B}_l$ and its corresponding asymptotic expression are given in the following proposition.
\begin{proposition} \label{prop_Cauchy}
	The definition of the Cauchy transform $\mc{G}_{\mb{\Xi},\mb{B}_{l}}(z)$, with $z\in\mbb{C}^+$, is given as follows
	\begin{equation}\label{equ_asymCauchy}
		\begin{aligned}
			\mc{G}_{\mb{\Xi},\mb{B}_{l}}(z)&=\frac{1}{T_t}\mbb{E}\left\{\mathrm{Tr}\left\{\mb{\Xi}(z\mb{I}_{T_t}-\mb{B}_l)^{-1}\right\}\right\} \\
			&\to_{a.s.}\frac{1}{T_t}\mathrm{Tr}\left(\mb{\Xi}\mc{\wt B}_{l}(z)\right),
		\end{aligned}
	\end{equation}
	where $\mb{\Xi}\in\mbb{C}^{T_t\times T_t}$ is a nonnegative definite matrix with uniformly bounded spectral norm and
	\begin{equation}
		\mc{\wt B}_{l}(z)=\left(\bm{\Upsilon}(z)- \mb{\wh W}^\dagger \mb{\ob F}_l^\dagger\mb{\Omega}^{-1}(z)\mb{\ob F}_l\mb{\wh W}\right)^{-1},
	\end{equation}
	the matrix ${\bf{\Omega}}(z)$ is defined as
	\begin{equation} \label{equ_vPi}
		\mb{\Omega}(z)=\bm{\wt \Upsilon}(z)-\left(\bm{\Gamma}(z)-\mb{\wh \Theta}^\dagger\mb{\ob G}_l^{\dagger}\bm{\wt \Gamma}^{-1}(z)\mb{\ob G}_l\mb{\wh \Theta}\right)^{-1}.
	\end{equation}
	
	The matrix-valued functions $\bm{\Upsilon}(z)$, $\bm{\Gamma}(z)$, $\bm{\wt \Gamma}(z)$ and $\bm{\wt \Upsilon}(z)$ satisfy the following fixed-point equations
	\begin{align}
		\bm{\Upsilon}(z) &= z\mb{I}_{T_t}-\mb{\wh W}^\dagger\Big(\zeta_{0l}(\mc{G}_{\mb{X}_{40}}(z)) + \sum_{k=1}^{K}\zeta_{k}(\mc{G}_{\mb{X}_{4k}}(z))\Big) \mb{\wh W}, \notag \\
		\bm{\Gamma}(z)&=-\mb{\wh \Theta}^\dagger\operatorname{blkdiag}\left\{\mb{0}_R,\eta_{1l}(\mc{G}_{\mb{X}_3}(z)),\dots,\eta_{Kl}(\mc{G}_{\mb{X}_3}(z))\right\}\mb{\wh \Theta}, \notag\\
		\bm{\wt \Gamma}(z)&=\mb{I}_{R}-\sum_{k=1}^{K}\widetilde{\eta}_{kl}\left(\mb{\Theta}_k\mc{G}_{\mb{X}_{2k}}(z)\mb{\Theta}_k^\dagger\right), \notag\\
		\bm{\wt \Upsilon}(z)&=-\operatorname{blkdiag}\Big\{\wt\zeta_{0l}({\mb{\wh W}\mc{G}_{\mb{X}_1}(z)\mb{\wh W}}^\dagger), \notag\\
		&\quad\wt\zeta_{1}({\mb{\wh W}\mc{G}_{\mb{X}_1}(z)\mb{\wh W}}^\dagger),\dots,\wt\zeta_{K}({\mb{\wh W}\mc{G}_{\mb{X}_1}(z)\mb{\wh W}}^\dagger)\Big\} \notag,
	\end{align}
	where the matrices $\eta_{kl}(\widetilde{\mb{D}})$, $\widetilde{\eta}_{kl}({\mb{D}})$, $\zeta_{0l}(\mb{Z}_0)$, $\widetilde{\zeta}_{0l}({\mb{\wt Z}_0})$, $\zeta_{k}(\mb{Z})$ and $\widetilde{\zeta}_{k}(\widetilde{\mb{Z}})$ are given by (\ref{equ_eta})-(\ref{equ_tildezeta}) in Appendix~\ref{appen_Cauchy}. The equations $\mc{G}_{\mb{X}_1}(z)$, $\mc{G}_{\mb{X}_{2k}}(z)$, $\mc{G}_{\mb{X}_3}(z)$ and $\mc{G}_{\mb{X}_{4k}}(z)$ are given by
	\begin{align}  
		\mc{G}_{\mb{X}_1}(z) &= \mc{\wt B}_{l}(z), \notag\\
		\mc{G}_{\mb{X}_{2k}}(z) &= \Big\{\Big(\bm{\Gamma}(z)-\mb{\wh \Theta}^\dagger\mb{\ob G}_l^{\dagger}\bm{\wt \Gamma}^{-1}(z)\mb{\ob G}_l\mb{\wh \Theta} \notag\\
		-&\left(\bm{\wt \Upsilon}(z)-\mb{\ob F}_l\mb{\wh W}\bm{\Upsilon}^{-1}(z)\mb{\wh W}^\dagger \mb{\ob F}_l^\dagger\right)^{-1}\Big)^{-1}\Big\}_{k+1},1\leq k\leq K, \notag\\
		\mc{G}_{\mb{X}_3}(z) &= \Big(\bm{\wt \Gamma}(z)-\mb{\ob G}_l\mb{\wh \Theta}\Big(\bm{\Gamma}(z) \notag\\
		&-\left(\bm{\wt \Upsilon}(z)-\mb{\ob F}_l\mb{\wh W}\bm{\Upsilon}^{-1}(z)\mb{\wh W}^\dagger \mb{\ob F}_l^\dagger\right)^{-1}\Big)^{-1}\mb{\wh \Theta}^\dagger\mb{\ob G}_l^{\dagger}\Big)^{-1},\notag \\
		\mc{G}_{\mb{X}_{4k}}(z) &= \Big\{\Big(\bm{\wt \Upsilon}(z)-\mb{\ob F}_l\mb{\wh W}\bm{\Upsilon}^{-1}(z)\mb{\wh W}^\dagger \mb{\ob F}_l^\dagger\notag\\
		-&\left(\bm{\Gamma}(z)-\mb{\wh \Theta}^\dagger\mb{\ob G}_l^{\dagger}\bm{\wt \Gamma}^{-1}(z)\mb{\ob G}_l\mb{\wh \Theta}\right)^{-1}\Big)^{-1}\Big\}_{k+1},0\leq k\leq K,\notag
	\end{align}
	where the notation $\{\mb A\}_{k+1}$ with $\mb A\in \mbb{C}^{L_{AR}\times L_{AR}}$ denotes the $(k+1)$-th diagonal matrix block containing entries from $\sum_{i=0}^{k-1}L_i+1$ to $\sum_{i=0}^{k}L_i$ rows and columns of $\mb{A}$.
\end{proposition}
\begin{proof}
	The definition and derivation of the Cauchy transform $\mc{G}_{\mb{\Xi},\mb{B}_{l}}(z)$ are given in~Appendix~\ref{appen_Cauchy}.
\end{proof}

Based on Proposition~\ref{prop_Cauchy}, the Cauchy transform $\mc{G}_{\mb{\Xi},\mb{B}_{l}}(z)$ can be obtained by the statistical CSI, including the mean and correlation matrices of each channel. Therefore, by invoking the Cauchy transform~(\ref{equ_asymCauchy}), the equation (\ref{equ_TrE}a) almost surely converges to 
\begin{equation}\label{equ_asym1}
	\begin{aligned}
		\operatorname{Tr}\left(\mb{W}_n^\dag\mbb{E}(\mb{\Phi}_{nn})^\dagger\mb{A}_n\right)\to_{a.s.}\operatorname{Tr}\left(\mb{\Psi}_n^\dagger\mb{A}_n\right),
	\end{aligned}
\end{equation}
where the block matrix $\mb{\Psi}_n\in\mbb{C}^{LT\times T}$ with $z=-\sigma^2$ is
\begin{equation} \label{equ_Psi}
	\mb{\Psi}_n=\left[
	\begin{array} {c}
		\big\{\mb{I}_{T_t}+\sigma^2\mc{\wt B}_{1}(z)\big\}_{nn}\\
		\vdots\\
		\big\{\mb{I}_{T_t}+\sigma^2\mc{\wt B}_{L}(z)\big\}_{nn}\\
	\end{array}\right].
\end{equation}

Following a similar process, the equation (\ref{equ_TrE}b) almost surely converges to 
\begin{equation}\label{equ_asym2}
	\operatorname{Tr}\left(\mb{A}_n^\dagger\mbb{E}(\mb{\Phi}_{nn})\mb{W}_n\right)\to_{a.s.}\operatorname{Tr}\left(\mb{A}_n^\dagger\mb{\Psi}_n\right).
\end{equation}

The equation (\ref{equ_TrE}c) can be written as
\begin{equation}
	\begin{aligned}
		&\quad\operatorname{Tr}\left(\mb{A}_n^\dagger\mbb{E}\Big\{\sum_{m=1}^{N}\mb{\Phi}_{nm}\mb{W}_m\mb{W}_m^\dagger\mb{\Phi}_{nm}^\dagger+\sigma^2\mb{S}_n\Big\}\mb{A}_n\right) \\ 
		&=\mbb{E}\left\{\operatorname{Tr}\left(\mb{A}_n^\dagger\mb{Q}_n\mb{A}_n\right)\right\},
	\end{aligned}
\end{equation}
where $\mb{Q}_n=\mb{\Phi}_n\mb{\wh{W}}\mb{\wh{W}}^\dagger\mb{\Phi}_n^\dagger+\sigma^2\mb{S}_n\in\mbb{C}^{LT\times LT}$ and $\mb{\Phi}_n=[\mb{\Phi}_{n1},\dots,\mb{\Phi}_{nN}]\in\mbb{C}^{LT\times T_t}$. The expression of the $l$-th $T\times T$ diagonal submatrix block of $\mb{Q}_n$ is given by
\begin{equation} \notag
	\begin{aligned}
		\left\{\mb{Q}_n\right\}_{ll}&=\sum_{m=1}^{N}\mb{U}_{nl}^\dagger\mb{H}_{ml}\mb{W}_m\mb{W}_m^\dagger\mb{H}_{ml}^\dagger\mb{U}_{nl}+\sigma^2\mb{U}_{nl}^\dag\mb{U}_{nl}\\
		&=\left\{\mb{I}_{T_t}-\sigma^2(\sigma^2\mb{I}_{T_t}+\mb{B}_l)^{-1}\right\}_{nn}.
	\end{aligned}
\end{equation}
The expression of the $(l,q)$-th $T\times T$ submatrix block of $\mb{Q}_n$ is given by
\begin{equation}\notag
	\begin{aligned}
		&\left\{\mb{Q}_n\right\}_{lq}=\sum_{m=1}^{N}\mb{U}_{nl}^\dagger\mb{H}_{ml}\mb{W}_m\mb{W}_m^\dagger\mb{H}_{mq}^\dagger\mb{U}_{nq}\\
		&=\left\{\left(\mb{I}_{T_t}-\sigma^2(\sigma^2\mb{I}_{T_t}+\mb{B}_l)^{-1}\right)\left(\mb{I}_{T_t}-\sigma^2(\sigma^2\mb{I}_{T_t}+\mb{B}_q)^{-1}\right)\right\}_{nn}.
	\end{aligned}
\end{equation}
Therefore, the equation (\ref{equ_TrE}c) can be expressed as 
\begin{align}
		\mbb{E}\left\{\operatorname{Tr}\left(\mb{A}_n^\dagger\mb{Q}_n\mb{A}_n\right)\right\}&=\sum_{l=1}^{L}\sum_{q=1}^{L}\mbb{E}\left\{\operatorname{Tr}\left(\mb{A}_{nl}^\dagger[\mb{Q}_n]_{lq}\mb{A}_{nq}\right)\right\} \notag \\
		&\to_{a.s.} \operatorname{Tr}\left\{\mb{A}_n^\dagger\mb{\wt Q}_n\mb{A}_n\right\}, \label{equ_AQA}
\end{align}
where the block matrix $\mb{\wt Q}_n\in\mbb{C}^{LT\times LT}$, whose $(l,q)$-th submatrix block is 
\begin{equation} \label{equ_tildeQ}
	\left\{\mb{\wt Q}_n\right\}_{lq}=
	\begin{aligned}
		&\left\{
		\begin{aligned}
			&\Big\{\mb{I}_{T_t}+\sigma^2\mc{\wt B}_{l}(z)\Big\}_{nn},  &&l=q,\\
			&\left\{(\mb{I}_{T_t}+\sigma^2\mc{\wt B}_{l}(z))(\mb{I}_{T_t}+\sigma^2\mc{\wt B}_{q}(z))\right\}_{nn},  &&l\neq q.
		\end{aligned}
		\right.
	\end{aligned}
\end{equation}

Therefore, the sum-MSE problem (\ref{equ_MSE}) can be asymptotically expressed as
\begin{align}
	\min_{\bm{\mc{A}}}\quad  
	&\sum_{n=1}^{N} \operatorname{Tr}(\mb{\wt E}_n),  \label{equ_eMMSE}
\end{align}
where 
\begin{equation} \label{equ_asymE}
		\operatorname{Tr}(\mb{\wt E}_n) = T-	\operatorname{Tr}(\mb{\Psi}_n^\dagger\mb{A}_n)-	\operatorname{Tr}(\mb{A}_n^\dagger\mb{\Psi}_n)
		+	\operatorname{Tr}\left(\mb{A}_n^\dagger\mb{\wt Q}_n\mb{A}_n\right).
\end{equation}

By minimizing the asymptotic sum-MSE problem (\ref{equ_eMMSE}), the asymptotic expression of the WCM $\mb{A}_n$ is given by
\begin{equation} \label{equ_A}
	\mb{A}_n^{(i+1)}=\mb{\wt Q}_n^{-1}\mb{\Psi}_n.
\end{equation}

\subsection{Optimization of $\mb V^{(i+1)}$ for given $\mb A^{(i+1)}$, $\mb W^{(i)}$ and $\mb{\Theta}^{(i)}$}
By substituting the asymptotic expressions (\ref{equ_asymE}) and (\ref{equ_A}) into (\ref{equ_WMMSE}) and fixing the other variables, the optimization problem with respect to $\mb{V}$ can be rewritten as
\begin{align}
	\min_{\bm{\mc{V}}}\quad  
	\sum_{n=1}^{N}\mu_n \left(\operatorname{Tr}\left(\mb{V}_n(\mb{I}_T-\mb{\Psi}_n^\dagger\mb{\wt Q}_n^{-1}\mb{\Psi}_n)\right)-\operatorname{log}\operatorname{det(\mb{V}_n)}\right).  \label{equ_equWMMSE}
\end{align}

 By taking the first-order derivative of (\ref{equ_equWMMSE}), the asymptotic expression of the auxiliary weight matrix $\mb V_n$ is given by
\begin{equation} \label{equ_optV}
	{\mb V}_n^{(i+1)} = \left(\mb{I}_T-\mb{\Psi}_n^\dagger\mb{\wt Q}_n^{-1}\mb{\Psi}_n\right)^{-1}.
\end{equation}


\subsection{Optimization of $\mb W^{(i+1)}$ for given $\mb A^{(i+1)}$, ${\mb V}^{(i+1)}$ and $\mb{\Theta}^{(i)}$}
Recall that the optimal TPMs can be obtained based on the Lagrange function (\ref{equ_eqL}), thus we aim to obtain the asymptotic expression of (\ref{equ_eqL}). In the ($i+1$)-th iteration, by invoking the Cauchy transform of the Gram matrix $\mb{B}_l$, the equation (\ref{equ_eqL}a) almost surely converges to  
\begin{equation}
	\begin{aligned}
		&\mathbb{E}\left\{\operatorname{Tr}\left(\mb{V}_n(\mb{W}_n^{(i+1)})^\dagger\mb{\Phi}_{nn}^\dagger\mb{A}_n\right)\right\}\\
		&\to_{a.s.} \operatorname{Tr}\left(\mb{V}_n(\mb{W}_n^{(i+1)})^\dagger((\mb{W}_n^{(i)})^\dagger)^{-1}\mb{\Psi}_n^\dagger\mb{A}_n\right).
	\end{aligned}
\end{equation}

Similarly, the equation (\ref{equ_eqL}b) can be rewritten as 
\begin{equation}
	\begin{aligned}
		&\mathbb{E}\left\{\operatorname{Tr}\left(\mb{V}_n\mb{A}_n^\dagger\mb{\Phi}_{nn}\mb{W}_n^{(i+1)}\right)\right\}\\
		&\to_{a.s.} \operatorname{Tr}\left((\mb{W}_n^{(i)})^{-1}\mb{W}_n^{(i+1)}\mb{V}_n\mb{A}_n^\dagger\mb{\Psi}_n\right).
	\end{aligned}
\end{equation}

Based on the matrix inversion lemma, the equation (\ref{equ_eqL}c) can be asymptotically expressed as 
\begin{equation}
	\begin{aligned}
		&\mathbb{E}\left\{\operatorname{Tr}\left((\mb{W}_n^{(i+1)})^\dagger\mb{\Phi}_{mn}^\dagger\mb{A}_m\mb{V}_m\mb{A}_m^\dagger\mb{\Phi}_{mn}\mb{W}_n^{(i+1)}\right)\right\} \\
		&\to_{a.s.}\sum_{l=1}^{L}\sum_{q=1}^{L}\operatorname{Tr}\Big\{(\mb{W}_n^{(i+1)})^\dagger((\mb{W}_n^{(i)})^\dagger)^{-1}\left\{\mb{I}_{T_t}+\sigma^2\mc{\wt B}_{l}(z)\right\}_{nm}\\
		&\qquad\qquad\mb{A}_{ml}\mb{V}_m\mb{A}_{mq}^\dagger\left\{\mb{I}_{T_t}+\sigma^2\mc{\wt B}_{q}(z)\right\}_{mn}(\mb{W}_n^{(i)})^{-1}\mb{W}_n^{(i+1)}\Big\}.
	\end{aligned}
\end{equation}

Therefore, the Lagrange function (\ref{equ_eqL}) almost surely converges to (\ref{equ_equL}), shown at the top of the next page. Then, the first-order optimal condition of $\wt L(\mb{W}_n^{(i+1)},\lambda_n)$ leads to the optimal asymptotic TPM of the $n$-th UE as (\ref{equ_W}), shown at the top of the next page.

\begin{figure*}[bt] 
	\begin{align} 
		\wt L(\mb{W}_n^{(i+1)},\lambda_n)& = \sum_{m=1}^{N}\mu_m\Big\{\operatorname{Tr}(\mb{V}_m)-\operatorname{Tr}\left(\mb{V}_m(\mb{W}_m^{(i+1)})^\dagger((\mb{W}_m^{(i)})^\dagger)^{-1}\mb{\Psi}_m^\dagger\mb{A}_m\right) -\operatorname{Tr}\left((\mb{W}_m^{(i)})^{-1}\mb{W}_m^{(i+1)}\mb{V}_m\mb{A}_m^\dagger\mb{\Psi}_m\right) \notag\\
		&+ \sum_{l=1}^{L}\sum_{q=1}^{L}\operatorname{Tr}\left\{(\mb{W}_n^{(i+1)})^\dagger((\mb{W}_n^{(i)})^\dagger)^{-1}\left\{\mb{I}_{T_t}+\sigma^2\mc{\wt B}_{l}(z)\right\}_{nm}\mb{A}_{ml}\mb{V}_m\mb{A}_{mq}^\dagger\left\{\mb{I}_{T_t}+\sigma^2\mc{\wt B}_{l}(z)\right\}_{mn}(\mb{W}_n^{(i)})^{-1}\mb{W}_n^{(i+1)}\right\}\Big\}\notag\\
		&+\lambda_n\left(\mathrm{Tr}(\mb{W}_{n}^{(i+1)}({\mb{W}_{n}^{(i+1)}})^\dag)-p_n\right) \label{equ_equL}.\\
		\mb{W}_n^{(i+1)}&=\mu_n\left(\sum_{m=1}^{N}\mu_m\sum_{l=1}^{L}\sum_{q=1}^{L}((\mb{W}_n^{(i)})^\dagger)^{-1}\left\{\mb{I}_{T_t}+\sigma^2\mc{\wt B}_{l}(z)\right\}_{nm}\mb{A}_{ml}\mb{V}_m\mb{A}_{mq}^\dagger\left\{\mb{I}_{T_t}+\sigma^2\mc{\wt B}_{q}(z)\right\}_{mn}(\mb{W}_n^{(i)})^{-1}+\lambda_n\mb{I}_{T}\right)^{-1} \notag\\
		&\qquad\qquad\qquad\qquad\times ((\mb{W}_n^{(i)})^\dagger)^{-1}\mb{\Psi}_n^\dagger\mb{A}_n\mb{V}_n.  \label{equ_W}\\
		\setcounter{equation}{51}
		\left\{\mb{\wt Q}^\prime_n\right\}_{lq}&=
		\begin{aligned}
			\left\{
			\begin{aligned}
				&\sigma^2\left\{\mc{\wt B}^\prime_{l}(z)\right\}_{nn},  &&l=q,\\
				&\left\{\sigma^2\mc{\wt B}_l^\prime(z)(\mb{I}_{T_t}+\sigma^2\mc{\wt B}_q(z))+\sigma^2(\mb{I}_{T_t}+\sigma^2\mc{\wt B}_l(z))\mc{\wt B}_q^\prime(z)\right\}_{nn},   &&l\neq q.
			\end{aligned}
			\right.
		\end{aligned} \label{equ_wtQP}
	\end{align}
	\hrulefill
\end{figure*}

\subsection{Optimization of $\mb{\Theta}^{(i+1)}$ for given $\mb A^{(i+1)}$, $\mb V^{(i+1)}$ and $\mb W^{(i+1)}$}
%
By substituting the asymptotic expressions (\ref{equ_asym1}), (\ref{equ_asym2}) and (\ref{equ_AQA}) into (\ref{equ_f}), the asymptotic expression of $f_n\left(\bm{\varTheta}\right)$ is 
\setcounter{equation}{48}
\begin{equation}
	\begin{aligned}
		\wh {f}_n(\bm{\varTheta})&=\operatorname{Tr}(\mb{V}_n\mb{A}_n^\dagger\mb{\wt Q}_n\mb{A}_n)-\operatorname{Tr}(\mb{V}_n\mb{\Psi}_n^\dagger\mb{A}_n)-	\operatorname{Tr}(\mb{V}_n\mb{A}_n^\dagger\mb{\Psi}_n),
	\end{aligned}
\end{equation}
where the expressions of the block matrices $\mb{\Psi}_n$ and $\mb{\wt Q}_n$ are given by (\ref{equ_Psi}) and (\ref{equ_tildeQ}) with replacing $\mb{W}^{(i)}$ with $\mb{W}^{(i+1)}$. Thus, we can rewrite the optimization problem (\ref{equ_optif}) as
\begin{align}
	\min_{\bm{\varTheta}^{(i+1)}}\quad  
	&\sum_{n=1}^{N}\mu_n \wh {f}_n(\bm{\varTheta}^{(i+1)})  \\
	~{\rm {s.t.}}~\quad     	    &(\ref{equ_WSR}\text{b}), \notag
\end{align}

Following a similar process as in Section~\ref{sec_PSM}, the optimal PSMs can be obtained through the Riemannian steepest descent algorithm. In the $(i + 1)$-th iteration, $\bm{\theta}^{(i+1)}$ can be updated by (\ref{equ_update_theta}) with  $\bm{\delta}^{(i)}=\operatorname{grad}_{\bm{\theta}^{(i)}}\wh{f}$, where ${\wh f}=\sum_{n=1}^{N}\mu_n {\wh f}_n$. The gradient $\operatorname{grad}_{\bm{\theta}}\wh{f}=\nabla \wh{f} -\operatorname{Re}\left\{\nabla \wh{f}\odot\bm{\theta}^\star\right\}\odot\bm{\theta}$ with $\nabla \wh{f}(\bm{\theta})=2\frac{d\wh{f}}{d\bm{\theta}^\star}$, where
\begin{equation} \notag
	\begin{aligned}
		\frac{d\wh{f}}{d\bm{\theta}^\star} 
		= \left[\frac{d\wh{f}}{d\bm{\theta}_{1}^\star},\dots,\frac{d\wh{f}}{d\bm{\theta}_{L_R}^\star}\right]^{\rm T}
		=\left[\frac{d\wh{f}}{d\mb{\Theta}_{1,1}^\star},\dots,\frac{d\wh{f}}{d\mb{\Theta}_{K,L_K}^\star}\right]^{\rm T},
	\end{aligned}
\end{equation}
where $\frac{d\wh{f}}{d\mb{\Theta}_{k,l_k}^\star}$ is given in the following proposition.
\begin{proposition}
	The derivative of $\wh{f}(\bm{\Theta})$ with respect to $\mb{\Theta}_{k,l_k}^\star$ can be expressed as $\frac{d\wh{f}}{d\mb{\Theta}_{k,l_k}^\star}=\sum_{n=1}^{N}\mu_n \frac{d\wh {f}_n(\mb{\Theta})}{d\mb{\Theta}_{k,l_k}^\star}$, where
	\begin{equation}
		\begin{aligned}
			\frac{d\wh {f}_n(\mb{\Theta})}{d\mb{\Theta}_{k,l_k}^\star} &=\operatorname{Tr}\left(\mb{A}_n\mb{V}_n\mb{A}_n^\dagger\mb{\wt Q}^\prime_n\right) -\sum_{l=1}^{L}\operatorname{Tr}\left(\sigma^2\mb{A}_{nl}\mb{V}_n[\mc{\wt B}^\prime_{l}(z)]_{nn}\right)\\
			&\qquad-\sum_{l=1}^{L}\operatorname{Tr}\left(\sigma^2\mb{V}_n\mb{A}^\dag_{nl}[\mc{\wt B}^\prime_{l}(z)]_{nn}\right)
		\end{aligned}
	\end{equation} 
	where we denote $\mb{\wt Q}_n^\prime$ as $\frac{d\mb{\wt Q}_n}{d\mb{\Theta}_{k,l_k}^\star}$ and $\mc{\wt B}_l^\prime(z)$ as $\frac{d\mc{\wt B}_l(z)}{d\mb{\Theta}_{k,l_k}^\star}$ for notation convenience. 
	The $(l,q)$-th block sub-matrices in $\mb{\wt Q}^\prime_n$ of dimension $T\times T$ are given by (\ref{equ_wtQP}), shown at the top of this page. The expression of $\mc{\wt B}_l^\prime(z)$ is given in Appendix~\ref{appen_derivative_II}.
\end{proposition}

\subsection{Overall Algorithm} \label{subsec_AO}
In the previous subsections, the asymptotic expressions of the WCM $\bm{\mc A}$, TPMs $\bm{\mc W}$ and PSMs $\bm{\varTheta}$ are obtained based on the operator-valued Cauchy transform. The proposed asymptotic AO algorithm only depends on the statistical CSI, which significantly reduces the interactive overhead. Specifically, in each realization of the APs/UEs location (the statistical CSI remains unchanged), the APs only need to feedback the statistical CSI once at the beginning of AO algorithm.
The proposed asymptotic AO algorithm is summarized in~Algorithm~\ref{alg_AO}. 
The convergence and complexity analyses are given as follows.

\textbf{Complexity Analysis:}
Assuming that the number of iterations of the fixed point equations in Proposition~1 is $I_{1}$ and ignoring the non-leading order terms, the complexity of Proposition~\ref{prop_Cauchy} is $O\left(I_{1}L\left(T_t^3+L_{AR}^3+R^3\right)\right)$. Through the closed-form expressions of the WCM $\mb{A}$, the auxiliary weight matrix $\mb{V}$ and the TPMs $\mb{W}$, the complexity of updating $\mb{A}$, $\mb{V}$ and $\mb{W}$ is given by $O\left(N(L^3T^3+T_t^3)\right)$.
Assuming that the number of iterations of the Riemannian steepest descent algorithm and the gradient in the PSM sub-problem are respectively $I_{2}$ and $I_3$, the complexity of updating PSMs is $O\left(I_{2}\left(L_RL^3T^3+I_{3}L(T_t^3+L_{AR}^3+R^3)\right)\right)$. Therefore, assuming that the number of iterations for the AO algorithm is $I_0$, the complexity in total is given by $O\Big(I_0\big((I_1L+N+I_2I_3L)T_t^3+(I_1L+I_2I_3L)L_{AR}^3+(I_1L+I_2I_3L)R^3+(N+I_2L_R)L^3T^3\big)\Big)$.

\textbf{Convergence Analysis:}
In Algorithm \ref{prop_Cauchy}, the block coordinate descent method is adopted, where four blocks, i.e., $\{\mb{A},\mb{V},\mb{W},\mb{\Theta}\}$ are optimized alternately, while keeping the other three blocks of variables fixed. We note that the variables $\{\mb{A},\mb{V},\mb{W}\}$ are updated with optimality in each iteration. In addition, since the Riemannian steepest descent algorithm is utilized in the optimization of the PSMs $\mb{\Theta}$, the search direction in the Riemannian manifold is ensured to be a descent direction. Therefore, the objective value of the weighted sum rate problem is non-decreasing. Since the weighted sum rate is upper-bounded by a finite value, the proposed AO algorithm is guaranteed to converge.

\RestyleAlgo{ruled}
\begin{algorithm}[h]
	\caption{Asymptotic AO Algorithm}\label{alg_AO}
	\renewcommand{\KwResult}{\textbf{Initiation: }}
	\KwIn{Channel statistics $\left\{\overline{\mb{F}},\mb{T},\mb{P},\mb{M},\mb{\ob G},\mb{R},\mb{C},\mb{N}\right\}$.}
	\KwOut{WCM $\bm{\mc A}$, TPMs $\bm{\mc W}$ and PSMs $\bm{\varTheta}$.}
	\KwResult{Maximum iteration number $I_{\rm max}$ and $i=0$. The UEs' TPM $\mb{W}^{0}_{n}=\sqrt{\frac{p_n}{T}}\mb{I}_{T},\forall n$. The RISs' PSM $\mb{\Theta}^{0}_{k}=\mb{I}_{L_k},\forall k$. The APs' local detector $\mb{U}$ based on $\mb{W}^{0}$ and $\mb{\Theta}^{0}$.} \\
	\Repeat{$i\geq I_{\rm max}$}
	{	
		Update the WCM $\mb{A}^{(i+1)}$ with $\mb{W}^{(i)}$ and $\mb{\Theta}^{(i)}$ via (\ref{equ_A}).\\
		Update the weight matrix $\mb{V}^{(i+1)}$  with $\mb{A}^{(i+1)}$, $\mb{W}^{(i)}$ and $\mb{\Theta}^{(i)}$ via (\ref{equ_optV}).\\
		Update UEs' TPMs $\mb{W}^{(i+1)}$ with $\mb{V}^{(i+1)}$, $\mb{A}^{(i+1)}$ and $\mb{\Theta}^{(i)}$ via (\ref{equ_W}).\\
		Update RISs' PSMs $\mb{\Theta}^{(i+1)}$ with $\mb{A}^{(i+1)}$, $\mb{V}^{(i+1)}$ and $\mb{W}^{(i+1)}$ via (\ref{equ_update_theta}).\\
		$i=i+1$
	}
\end{algorithm}

\section{Numerical Results} \label{sec_Result}
In this section, we provide comprehensive simulation results to verify the effectiveness of the proposed AO algorithm and the accuracy of the asymptotic AO algorithm based on the operator-valued free probability. 
\setcounter{equation}{52}
\subsection{Simulation Setup}
We assume that the APs and UEs are randomly distributed within a $1\times 1\text{km}^2$ square area. Following a similar setup as in~\cite{zhang2023performance}, the RISs are assumed to be randomly deployed inside circles of a radius of $10$m centered by the UEs. The transmit power of each UE is set to $p_k=23\text{dBm}$ and the noise power at each AP is $\sigma^2=-94\text{dBm}$. Without loss of generality, the priority $\mu_k$ are set equally for all UEs\footnote{The priorities can be adjusted to meet different system requirements, such as quality of service (QoS) or fairness among UEs.}. 
Following the argument in~\cite{dovelos2020massive}, the full correlation matrix of the direct channel between the $l$-th AP and the $n$-th UE is
\begin{equation} \notag
	\begin{aligned}
		\bm{\mc{C}}_{0,nl}& \triangleq\mathbb{E}\left\{\mathrm{vec}(\widetilde{\mb{F}}_{0,nl})\mathrm{vec}(\widetilde{\mb{F}}_{0,nl})^\dagger\right\}  \\
		&=(\mb{P}_{0,nl}\otimes\mb{T}_{0,nl})\mathrm{diag}(\mathrm{vec}(\mb{\wt M}_{0,nl}))(\mb{P}_{0,nl}\otimes\mb{T}_{0,nl})^\dagger,
	\end{aligned}
\end{equation}
where $\mb{\wt M}_{0,nl} = \mb{M}_{0,nl}\otimes\mb{M}_{0,nl}$ is the eigenmode coupling matrix. In the direct channel (\ref{equ_F0}), the Rician factor $\kappa_{nl}$ and the large-scale fading coefficient $\beta_{nl}$ are incorporated into the channel model through the constraints $\lVert\overline{\mb{F}}_{0,nl}\rVert_F^2=R\beta_{nl}\frac{\kappa_{nl}}{\kappa_{nl}+1}$ and $\operatorname{Tr}(\bm{\mc{C}}_{0,nl})=R\beta_{nl}\frac{1}{\kappa_{nl}+1}$. A similar normalization can be also applied to the links between APs and RISs and between RISs and UEs. 
Consider the uniform planar array (UPA) equipped with all nodes, the LoS channel model between the $n$-th UE and the $l$-th AP is given as follows
\begin{align}
	\mb{\ob F}_{0,nl} = \mb{a}_{r}(\theta_{nl}^{r},\phi_{nl}^{r})\mb{a}_{t}(\theta_{nl}^{t},\phi_{nl}^{t})^\dagger,
\end{align}
where $\theta_{nl}^{r}(\phi_{nl}^{r})$ and $\theta_{nl}^{t}(\phi_{nl}^{t})$ are respectively the azimuth (elevation) angles of arrival and departure (AoA and AoD) of the LoS path between the $n$-th UE and the $l$-th AP. The receive and transmit array response vectors are  $\mb{a}_{r}(\theta,\phi)=\mb{a}_{rh}(\theta,\phi)\otimes\mb{a}_{rv}(\phi)$ and  $\mb{a}_{t}(\theta,\phi)=\mb{a}_{th}(\theta,\phi)\otimes\mb{a}_{tv}(\phi)$, respectively, where
\begin{align}
	\mb{a}_{rh}(\theta,\phi)&=\left[1,e^{j2\pi\delta\sin(\phi)\sin(\theta)},\dots,e^{j2\pi(N^R_H-1)\delta\sin(\phi)\sin(\theta)}\right]^{\rm T}, \notag\\
	\mb{a}_{rv}(\phi)&=\left[1,e^{j2\pi\delta\cos(\phi)},\dots,e^{j2\pi(N^R_V-1)\delta\cos(\phi)}\right]^{\rm T},\notag
\end{align}
where $R=N^R_HN^R_V$ with $N^R_H$ and $N^R_V$ being the number of antennas in horizontal dimension and vertical dimension at receiving sides and $\delta$ is the normalized antenna spacing. The transmit array response vector has a similar structure as $\mb{a}_{r}(\theta,\phi)$ and is omitted. The LoS connections between AP and RIS, and RIS and UE can be modeled with a similar approach.

 The large-scale fading coefficients of the direct channels between every two devices can be expressed as $\beta=-30-\alpha\operatorname{log}(d)$, where $\alpha$ is the path loss exponent and $d$ is the distance between two devices. As discussed in \cite{zhang2022active}, the direct links between APs and UEs are weak due to severe obstruction, while the links between APs and RISs and between RISs and UEs are strong since the RISs are universally deployed in high-rise building facades. Thus, the path loss exponents of AP-UE, AP-RIS and RIS-UE channels are set as 3.8, 2 and 2.2, respectively. In addition, unless otherwise stated, the Rician factors of the direct links between AP-UE, AP-RIS and RIS-UE are set as 3, 10 and 10, respectively. Furthermore, the correlation matrices and variance profiles of the channel are generated randomly but fixed in each Monte Carlo simulation. In each case, we annotate 'opt' on the superscript of variables to indicate the optimized variables. On the other hand, the unoptimized variables are denoted as $\mb{I}$, such as $\mb{A}=\mb{I}$, $\mb{W}=\mb{I}$ and $\mb{\Theta}=\mb{I}$ respectively denote that the WCM, TPMs and PSMs are set as $\mb{A}_{nl}=\frac{1}{L}\mb{I}_{T}$, $\mb{W}_n=\sqrt{\frac{p_n}{T}}\mb{I}_{T}$ and $\mb{\Theta}_l=\mb{I}_{L_k}$. Unless otherwise stated, the WCM $\mb{A}$ is optimized by default and the notation $\mb{A}=\mb{A}^{\rm opt}$ may be omitted.
 
To demonstrate the effectiveness of the proposed algorithm, we consider the following benchmarks
 \begin{itemize}
 	\item \textbf{Fully centralized processing (FCP):} The design of combining and precoding matrices is performed at the CPU with instantaneous CSI transmitted from APs~\cite{wang2023uplink}. With RIS-involved cases, the RIS's phase-shifts can be updated by the gradient descent algorithm.
 	\item \textbf{LSFD-MMSE/MR:} Each AP first detect the uplink signals with MMSE/MR detector based on local instantaneous CSI, and the CPU performs LSFD detection with expectations of channel matrices~\cite{wang2023uplink}.
 \end{itemize}

\subsection{Simulation Results}

	In Fig.~\ref{Fig_Accuracy}, we first demonstrate the accuracy of asymptotic expressions of the WCM, TPMs and PSMs in the asymptotic AO algorithm in different system setups. The solid lines represent the theoretical results obtained based on the asymptotic AO algorithm in Section~\ref{sec_asymAO} and the markers represent Monte Carlo simulation results based on the AO algorithm in Section~\ref{sec_AO}. As shown in Fig.~\ref{Fig_Accuracy}, the theoretical results fit the Monte Carlo simulation results accurately regardless of joint optimization or separate optimization of each variable, which illustrates the accuracy of the proposed algorithm. In addition, Fig.~\ref{Fig_Accuracy} shows that the sum-rate monotonously increases along with the number of APs, which is achieved by the macro-diversity gain. Note that the performance gain between the optimized cases (blue line) and the unoptimized cases (purple line) becomes large when the number of transmitting antennas and RIS's reflecting antennas grows, which indicates the effectiveness of the proposed algorithm.

The convergence of the proposed AO algorithm is illustrated in Fig.~\ref{Fig_Convergence}, where the numbers of UEs are set as $4$, $8$ and $12$, respectively. In each case, the sum rate increases with the number of iterations and converges to the maximum value within $5$ iterations, which shows a rapid rate of convergence for the proposed algorithm. Furthermore, compared to the initial state, where the WCM $\mb{A}$, TPMs $\mb{W}$ and PSMs $\mb{\Theta}$ are set to be identity matrices, the converged sum rate has increased by $74\%$, $48\%$ and $32\%$ at $N=4$, $N=8$ and $N=12$, respectively, which demonstrates that the proposed AO algorithm has great effectiveness to improve the sum rate. On the other hand, the theoretical results are consistent with simulation results in each iteration, which further demonstrates the accuracy of the proposed asymptotic AO algorithm.

\begin{figure}[tb] 
	\centerline{\includegraphics[width=0.8\columnwidth]{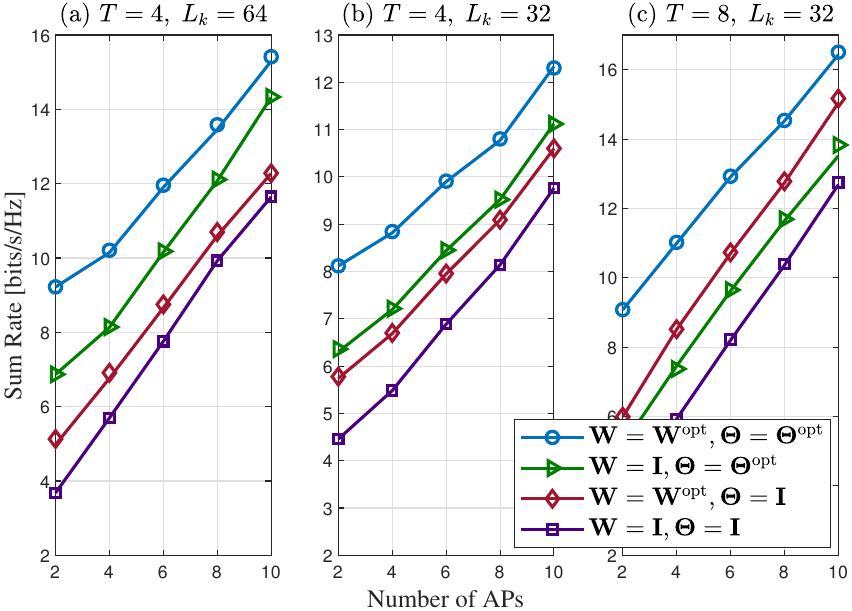}}
	\caption{Sum-rate versus the number of APs with varying numbers of RIS's passive reflecting antennas and transmitting antennas per UE, where $N=2$, $K=2$ and $R=8$. The solid lines represent theoretical results and the markers represent Monte Carlo simulation results.}
	\label{Fig_Accuracy}
\end{figure}

\begin{figure}[tb] 
	\centerline{\includegraphics[width=0.8\columnwidth]{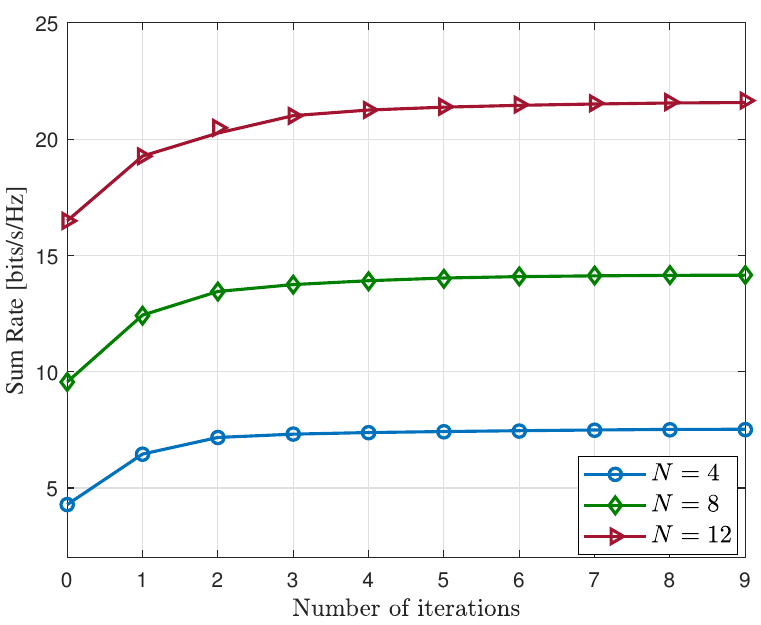}}
	\caption{Sum-rate versus the number of iterations with varying numbers of UEs, where $L=4$, $K=2$, $R=8$, $T=4$ and $L_k=32$ for $k=1,\dots,K$. The solid lines represent theoretical results and the markers represent Monte Carlo simulation results.}
	\label{Fig_Convergence}
\end{figure}

\begin{figure}[tb] 
	\centerline{\includegraphics[width=0.8\columnwidth]{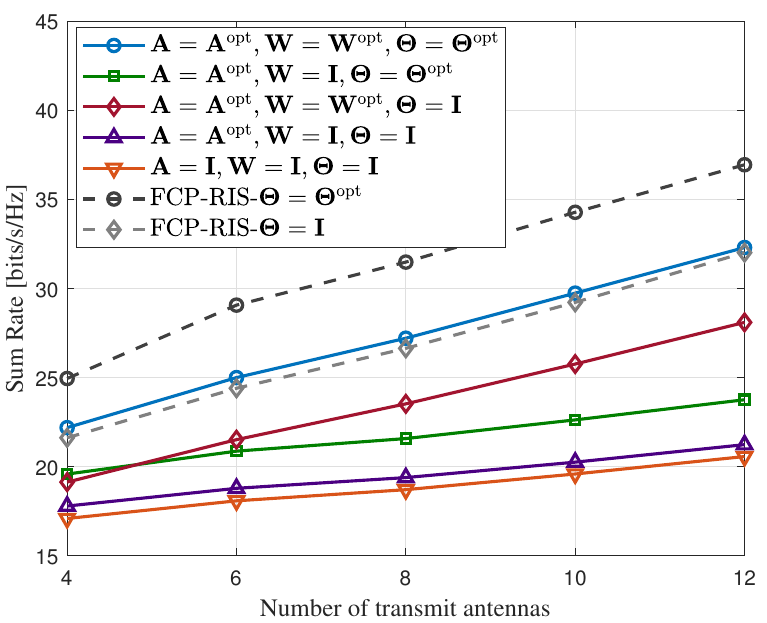}}
	\caption{Sum-rate versus the number of transmitting antennas, where $N=4$, $L=8$, $K=4$, $R=8$ and $L_k=32$ for $k=1,\dots,K$.}
	\label{Fig_NTx}
\end{figure}

To further illustrate the effectiveness of the proposed AO algorithm, we evaluate the performance of independent variable optimization within the AO algorithm in~Fig.~\ref{Fig_NTx}. Specifically, we consider the following cases: 1) Optimized case; 2) Unoptimized case; 3) Partial optimization case, i.e., only one or two variables are optimized. As shown in Fig.~\ref{Fig_NTx}, compared to the unoptimized case, the separate optimization of the WCM, TPMs and PSMs can efficiently improve the sum rate. 
In addition, note that the cases $\mb{A}=\mb{A}^{\rm opt}$, $\mb{W}=\mb{I}$, $\mb{\Theta}=\mb{I}$ and $\mb{A}=\mb{I}$, $\mb{W}=\mb{I}$, $\mb{\Theta}=\mb{I}$ respectively correspond to the LSFD-MMSE and LSFD-MR in~\cite{wang2022uplink}, but with RISs deployed. Compared to LSFD-MR, the proposed AO algorithm can improve the sum rate by at least $30.1\%$. Besides, compared to LSFD-MMSE, the individual optimization of the TPMs or the PSMs can improve the sum rate by at least $7.6\%$ and $10.2\%$, respectively, and the joint optimization of TPMs and PSMs can improve the sum-rate by at least $24.8\%$, which demonstrate the superiority of the proposed AO algorithm. Furthermore, the FCP cases with both optimized and unoptimized RIS's phase-shifts have a maximum performance gain of $14.3\%$ compared to the proposed AO algorithm, which concludes that the proposed AO algorithm can achieve comparable performance with lower interaction overhead.

\begin{figure}[tb] 
	\centerline{\includegraphics[width=0.8\columnwidth]{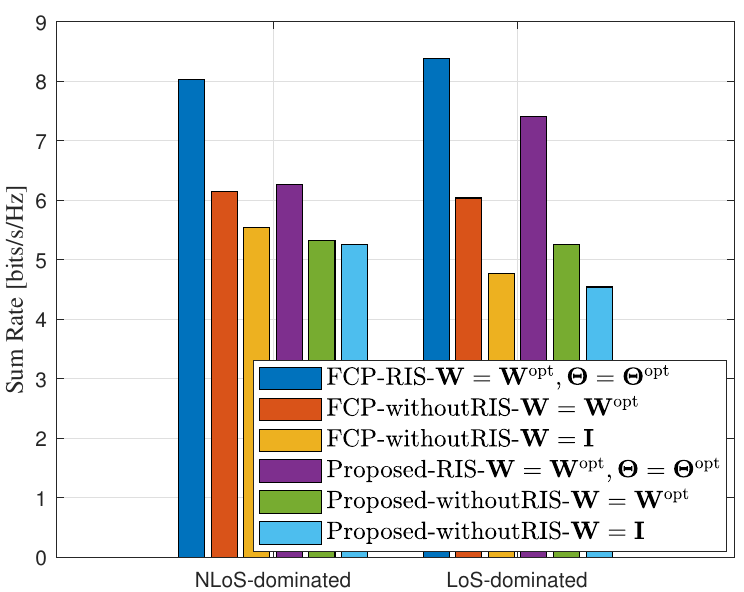}}
	\caption{Sum-rate comparison in the NLoS-dominated environment ($\kappa=0$) and the LoS-dominated environment ($\kappa=10$), where $N=4$, $L=4$, $K=2$, $R=4$, $T=2$ and $L_k=32$ for $k=1,\dots,K$.}
	\label{Fig_Rician}
\end{figure}

\begin{figure}[tb] 
	\centerline{\includegraphics[width=0.8\columnwidth]{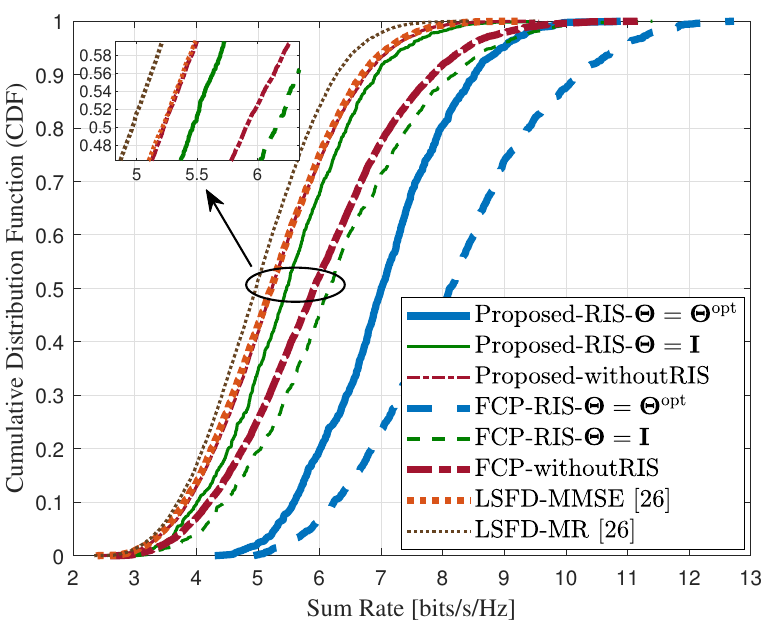}}
	\caption{Comparison with FCP and LSFD-MMSE/MR schemes in~\cite{wang2023uplink} with $N=4$, $L=4$, $R=4$ and $T=2$. In the RIS-assisted scenarios, the number of RISs and the passive reflecting elements per RISs are set as $K=2$ and $L_k=32$ for $k=1,\dots,K$.}
	\label{Fig_Compare}
\end{figure}

To demonstrate the benefit of the deployment of RISs in both the NLoS-dominated and the LoS-dominated environments, we compare the performances of the proposed algorithm and the FCP method with/without RIS deployment in the NLoS-dominated environment ($\kappa=0$) and the LoS-dominated environment ($\kappa=10$), as shown in Fig.~\ref{Fig_Rician}. In the NLoS-dominated environment, the proposed algorithm and the FCP method with optimized RIS can respectively achieve $18\%$ and $30\%$ performance gains than that without RIS deployment. In the LoS-dominated environment, the proposed algorithm and the FCP method with optimized RIS can respectively achieve $41\%$ and $40\%$ performance gains than that without RIS deployment. The results illustrate that the deployment of RISs can effectively improve the system performance in both the NLoS-dominated and the LoS-dominated environments. In addition, we note that the sum rate will reduce with the increase of the Rician factor in conventional CF networks without RIS deployment. However, the performance loss can be alleviated by deploying RISs or optimizing the transmitting precoder, which demonstrates the superiority of the proposed algorithm.


In Fig.~\ref{Fig_Compare}, we compare the proposed AO algorithm with the benchmarks in terms of cumulative distribution function (CDF) of the sum rate. 
The FCP with optimized RIS's phase-shifts has an average $16\%$ performance gain than that of the proposed algorithm. Furthermore, compared to the conventional CF networks, the FCP algorithm and proposed algorithm can respectively achieve $37\%$ and $34\%$ performance gains by deploying RISs. In addition, the proposed algorithm shows similar performance as the LSFD-MMSE scheme and has about $5\%$ performance gain than the LSFD-MR scheme. However, it should be mentioned that, in the LSFD-MMSE/MR schemes, the APs need to feedback the statistical information in each iteration of the AO algorithm to jointly design the combining matrix and the precoding matrix. Nonetheless, in the proposed AO algorithm, the APs only need to feedback the statistical CSI to the CPU once for each realization of the APs/UEs location, which significantly reduces the interaction overhead.  

\section{Conclusion} \label{sec_conclusion}
In this paper, we focus on the design of transceivers and phase-shifting matrices in multi-RIS-assisted CF networks. Specifically, the APs jointly serve multiple ground UEs with the assistance of multiple RISs and connect to the CPU to perform joint detection. To reduce the interaction overhead as well as exploit the macro-diversity, we propose a two-layer distributed reception scheme comprising the local MMSE detection at APs and weighted combining at the CPU. Aiming to maximize the weighted sum rate, we adopt the WMMSE framework to jointly design the WCM, TPMs and PSMs. Furthermore, considering the capacity-limited wireless fronthaul links, we resort to the operator-valued free probability to obtain the asymptotic expressions of the WCM, TPMs and PSMs. Since the asymptotic AO algorithm only depends on the statistical CSI, the APs only need to send the statistical CSI to the CPU once for the joint optimization, which greatly reduces the interaction overhead. Numerical results show that compared to the traditional CF networks without RIS deployment, the proposed AO algorithm can significantly improve the system performance as well as reduce the interaction overhead.


\begin{appendices}
\vspace{-0.2cm}

\section{Preliminary on Operator-Valued Free Probability Theory}
Let $\mc A$ be a unital algebra. A non-commutative probability space is defined as $(\mc A, \phi)$ consisting of $\mc A$ and a linear functional $\phi:\mc{A}\to \mbb{C}$. The elements of a non-commutative probability space are called non-commutative random variables.
Let $\mc{B}\subset \mc{A}$ be a unital subalgebra. For $\mc{H}\in \mc{A}$, a linear map $\mbb{E}_{\mc B}[\mc H]:\mc{A}\to\mc{B}$ is a $\mc{B}$-valued conditional expectation, if $\mbb{E}_{\mc B}[b]=b$ for all $b\in \mc{B}$ and $\mbb{E}_{\mc B}[b_1\mc{H}b_2]=b_1\mbb{E}_{\mc B}[\mc{H}]b_2$ for all $\mc{H}\in \mc{A}$ and $b_1,b_2\in \mc{B}$. Then, a $\mc{B}$-valued probability space is denoted as $(\mc{A},\mbb{E}_{\mc B},\mc{B})$, consisting of $\mc{B}\subset \mc{A}$ and the $\mc{B}$-valued conditional expectation $\mbb{E}_{\mc B}$. The elements of a $\mc B$-valued probability space are called $\mc B$-valued random variables.

Let $\mc{A}_1,\ldots,\mc{A}_K$ be the subalgebras of $\mc{A}$ with $\mc{B}\subset\mc{A}_k$ for all $1\le k\le K$. We also let $\{\bm{\mc{H}}_k\in\mc{A}_k, 1\le k\le K\}$ denote a family of operator-valued random variables, which are free with amalgamation over $\mc{B}$ according to the following definition.
\begin{definition}\label{defConvention}
	Let $n$ be an arbitrary integer. The families of random variables $\{\bm{\mc{H}}_1,\ldots,\bm{\mc{H}}_K\}$ are  free with amalgamation over $\mc{B}$, if for every family of index $\{k_1,\ldots,k_n\}\subset\{1,\ldots,K\}$ with $k_1\neq k_2$, \ldots, $k_{n-1}\neq k_n$, and every family of polynomials $\{P_1,\ldots,P_n\}$ satisfying $\mbb{E}_\mc{B}[P_j(\bm{\mc{H}}_{k_j})] = 0$, $j\in\{1,\ldots,n\}$, we have $\mbb{E}_{\mc{B}}\left[\prod_{j=1}^{n} P_j(\bm{\mc{H}}_{k_j})\right] = 0$.
\end{definition}

For a self-adjoint random variable $\mc{H}\in \mc{A}$ and $b\in\mbb{H}_{+}(\mc{B})$, where $\mbb{H}_{+}(\mc{B})$ is the operator upper half plane, defined by $\mbb{H}_{+}(\mc{B})=\left\{b\in \mc{B}:\Im(b)>0\right\}$. Then, the $\mc{B}$-valued Cauchy transform $\mc{G}^{\mc B}_{\mc H}(b)$ is defined as
\begin{equation}
	\mc{G}^{\mc B}_{\mc H}(b)=\mbb{E}_{\mc B}\left\{(b-\mc{H})^{-1}\right\}.
\end{equation}

Defining the operator lower half plane $\mbb{H}_{-}(\mc{B})$ as $\mbb{H}_{-}(\mc{B})=\left\{b\in \mc{B}:\Im(b)<0\right\}$. Then $\mc{B}$-valued $R$-transform of $\mc H$ is defined as
\begin{equation}
	R_{\mc H}^{\mc{B}}(b)=\sum_{n\geq 0}\kappa_{n+1}^{\mc{B}}\left({\mc H}b,\dots,{\mc H}b,{\mc H}b,{\mc H}\right),
\end{equation}
where $\kappa_{n}^{\mc{B}}$ is the $\mc{B}$-valued cumulants, which can be obtained from
the $\mc{B}$-valued moments.

The relation between $\mc{B}$-valued Cauchy transform and $R$-transform is given by
\begin{equation}
	R_{\mc H}^{\mc{B}}(b)=(\mc{G}^{\mc B}_{\mc H}(b))^{-1}-b^{-1},
\end{equation}
where $(\mc{G}^{\mc B}_{\mc H}(b))^{-1}: \mbb{H}_{-}(\mc{B})\to \mbb{H}_{+}(\mc{B})$ denotes the inverse function.

\section{Proof of Proposition~\ref{prop_Cauchy}} \label{appen_Cauchy}

Let $\mc{M} = \mb{M}_n(\mc{C})$ denote the algebra of $n\times n$ complex random matrices, where $n=T_t+R+2L_{AR}$. Consider the sub-algebra $\mc{D}\subset\mc{M}$ as the $n\times n$ block diagonal matrix. For $\mb{X}\in \mc{M}$, we define the linear map $\mbb{E}_{\mc D}[\mb{X}]:\mc{M}\to\mc{D}$. Then, we can define an operator-valued probability space $(\mc{M}, \mbb{E}_{\mc D}, \mc{D})$. For the $\mc{M}$-valued random variable $\mb{L}_{\mb{B}_l}\in (\mc{M},\mbb{E}_{\mc D},\mc{D})$, its $\mc{D}$-valued Cauchy transform is defined as
\begin{equation} \label{equ_GBl}
	\mc{G}_{\mb{L}_{\mb{B}_l}}^{\mc{D}}(\mb{\Lambda}(z)) = \mbb{E}_{\mc{D}}\left[\left(\mb{\Lambda}(z) - \mb{L}_{\mb{B}_l}\right)^{-1}\right],
\end{equation}
where $\mb{\Lambda}(z)\in\mc{M}$. 

To obtain the Cauchy transform, we divide the derivation into the following four steps. Recalling that the Gram matrix $\mb{B}=\mb{\wh {W}}^\dag\mb{F}^\dagger\mb{\wh{\Theta}}\mb{G}^\dagger\mb{G}\mb{\wh{\Theta}}\mb{F}\mb{\wh {W}}$ in the considered problem, where the subscripts are omitted. Note that both $\mb{F}$ and $\mb{G}$ are non-central and with non-trivial spatial correlations, and thus, are not free over the complex algebra in the classic free probability aspect. Thus, we construct a linearization matrix in \textbf{\textit{Step}~1}, whose operator-valued Cauchy transform is related to the Cauchy transform of $\mb{B}$. In specific, the equivalence relationship between the constructed operator-valued Cauchy transform and the Cauchy transform $\mc{G}_{\mb{\Xi},\mb{B}_{l}}(z)$ is established in \textbf{\textit{Step}~2}. In addition, the explicit expression of the constructed operator-valued Cauchy transform is obtained in \textbf{\textit{Step}~3}. Then, based on the equivalence relationship in \textbf{\textit{Step}~2}, the Cauchy transform $\mc{G}_{\mb{\Xi},\mb{B}_{l}}(z)$ is obtained in \textbf{\textit{Step}~4}.

\textbf{\textit{Step}~1: Linearization matrix construction}: Based on the Anderson’s linearization trick~\cite{belinschi2017analytic}, the random matrix $\mb{L}_{\mb{B}_l}$ of size $n\times n$ is
\begin{align} \label{equ_LB}
	\mb{L}_{\mb{B}_l}=
	\left[\begin{array}{cccc}
		\mb{0}_{T_t} & \mb{0}_{T_t\times L_{AR}} & \mb{0}_{T_t\times R} & \mb{\wh W}^\dagger \mb{F}_l^\dagger\\
		\mb{0}_{L_{AR}\times T_t} & \mb{0}_{L_{AR}} & \mb{\wh \Theta}^\dagger\mb{G}_l^{\dagger} & -\mb{I}_{L_{AR}}\\
		\mb{0}_{R\times T_t} & \mb{G}_l\mb{\wh \Theta} & -\mb{I}_{R} & \mb{0}_{R\times L_{AR}}\\
		\mb{F}_l\mb{\wh W} & -\mb{I}_{L_{AR}} & \mb{0}_{L_{AR}\times R} & \mb{0}_{L_{AR}}
	\end{array}\right].
\end{align}

The linearization matrix $\mb{L}_{\mb{B}_l}$ into the deterministic and random parts, i.e., $\mb{L}_{\mb{B}_l}=\ob{\mb{L}}_{\mb{B}_l}+\mb{\wt L}_{\mb{B}_l}$, whose expressions are given as follows
\begin{align}
	\mb{\wt L}_{\mb{B}_l} &=
	\left[\begin{array}{cccc}
		\mb{0} & \mb{0} & \mb{0} & \mb{\wh W}^\dagger \mb{\wt F}_l^\dagger\\
		\mb{0} & \mb{0} & \mb{\wh \Theta}^\dagger\mb{\wt G}_l^{\dagger} & \mb{0}\\
		\mb{0} & \mb{\wt G}_l\mb{\wh \Theta} & \mb{0} & \mb{0}\\
		\mb{\wt F}_l\mb{\wh W} & \mb{0} & \mb{0} & \mb{0}
	\end{array}\right], \\
	\mb{\ob L}_{\mb{B}_l} &=
	\left[\begin{array}{cccc}
		\mb{0} & \mb{0} & \mb{0} & \mb{\wh W}^\dagger \mb{\ob F}_l^\dagger\\
		\mb{0} & \mb{0} & \mb{\wh \Theta}^\dagger\mb{\ob G}_l^{\dagger} & -\mb{I}\\
		\mb{0} & \mb{\ob G}_l\mb{\wh \Theta} & -\mb{I} & \mb{0}\\
		\mb{\ob F}_l\mb{\wh W} & -\mb{I} & \mb{0} & \mb{0}
	\end{array}\right] \label{equ_barL},
\end{align}
where 
$\mb{\ob G}_l = [\mb{I}_{R},\mb{\ob G}_{1l},\dots,\mb{\ob G}_{Kl}]$,
$\mb{\wt G}_l = [\mb{0}_{R},\mb{\wt G}_{1l},\dots,\mb{\wt G}_{Kl}]$,
$\mb{\ob F}_l = [\mb{\ob F}_{0l}^\dag,\mb{\ob F}_1^\dag,\dots,\mb{\ob F}_K^\dag]^\dag$, 
$\mb{\wt F}_l = [\mb{\wt F}_{0l}^\dag,\mb{\wt F}_1^\dag,\dots,\mb{\wt F}_K^\dag]^\dag$, 
with $\mb{\wt F}_{0l}=[\mb{\wt F}_{0,1l},\dots,\mb{\wt F}_{0,Nl}]$, $\mb{\ob F}_{0l}=[\mb{\ob F}_{0,1l},\dots,\mb{\ob F}_{0,Nl}]$, $\mb{\wt F}_k=[\mb{\wt F}_{1k},\dots,\mb{\wt F}_{Nk}]$ and $\mb{\ob F}_k=[\mb{\ob F}_{1k},\dots,\mb{\ob F}_{Nk}]$. 

\textbf{\textit{Step}~2: Construct operator-valued Cauchy transform and establish equivalence relations}: The expectation $\mbb{E}_{\mc{D}}[\mb X]$ in (\ref{equ_GBl}) is defined as
\begin{equation} \label{equ_EDX}
	\mbb{E}_\mc{D}\left[\mb{X}\right] \!=\! 
	\left[\begin{array}{c:c:c:c}
		\mbb{E}\left[\mb{X}_1\right] & & & \\\hdashline
		& \mbb{E}\left[\mb{X}_2\right] & & \\\hdashline
		& & \mbb{E}\left[\mb{X}_3\right] & \\\hdashline
		& & &  \mbb{E}[\mb{X}_4]
	\end{array}\right],
\end{equation}
where $\mb{X}_1$ is a $T_t\times T_t$ sub-matrix and $\mb{X}_3$ is a $R\times R$ sub-matrix. The $L_{AR}\times L_{AR}$ block diagonal matrices $\mb{X}_2=\operatorname{blkdiag}\left\{\mb{X}_{20},\mb{X}_{21},\dots,\mb{X}_{2K}\right\}$ and $\mb{X}_4=\operatorname{blkdiag}\left\{\mb{X}_{40},\mb{X}_{41},\dots,\mb{X}_{4K}\right\}$, where $\mb{X}_{2k}$ and $\mb{X}_{4k}$ are $L_k\times L_k$ sub-matrices with $L_0=R$. Note that $\mc{G}_{\mb{L}_{\mb{B}_l}}^{\mc{D}}(\mb{\Lambda}_n(z))\in \mc{D}$, by the same matrix partitioning as in (\ref{equ_EDX}), we can define $\mc{G}_{\mb{L}_{\mb{B}_l}}^{\mc{D}}(\mb{\Lambda}_n(z))$ as 
\begin{equation} \label{equ_P_GBL}
	\mc{G}_{\mb{L}_{\mb{B}_l}}^{\mc{D}}(\mb{\Lambda}(z))=\operatorname{blkdiag}\left\{\mc{G}_{\mb{X}_1}(z),\mc{G}_{\mb{X}_2}(z),\mc{G}_{\mb{X}_3}(z),\mc{G}_{\mb{X}_4}(z)\right\},
\end{equation}
where $\mc{G}_{\mb{X}_2}(z)=\operatorname{blkdiag}\left\{\mb{0}_R,\mc{G}_{\mb{X}_{21}}(z),\dots,\mc{G}_{\mb{X}_{2K}}(z)\right\}$ and $\mc{G}_{\mb{X}_4}(z)=\operatorname{blkdiag}\left\{\mc{G}_{\mb{X}_{40}}(z),\mc{G}_{\mb{X}_{41}}(z),\dots,\mc{G}_{\mb{X}_{4K}}(z)\right\}$.

The matrix function $ \mb{\Lambda}(z) $ is denoted as
\begin{align} \label{equ_lambdaz}
	\mb{\Lambda}(z) = \begin{bmatrix}
		z\mb{I}_{T_t} & \mb{0}_{T_t\times (R+2L_{AR})}\\
		\mb{0}_{(R+2L_{AR})\times T_t} & \mb{0}_{(R+2L_{AR})\times (R+2L_{AR})} 
	\end{bmatrix}.
\end{align} 

By substituting (\ref{equ_LB}) and (\ref{equ_lambdaz}) into (\ref{equ_GBl}) and applying the matrix inversion lemma, we 
note that the upper-left $T_t\times T_t$ matrix block of (\ref{equ_GBl}) is equivalent to $(z\mb{I}_{T_t}-\mb{B}_l)^{-1}$ in $\mc{G}_{\mb{\Xi},\mb{B}_{l}}(z)$, thus we have
\begin{equation} \label{equ_eq_G}
	\mc{G}_{\mb{\Xi},\mb{B}_{l}}(z) = \frac{1}{T_t}\mathrm{Tr}\left(\mb{\Xi}\left\{\mc{G}_{\mb{L}_{\mb{B}_l}}^{\mc{D}}(\mb{\Lambda}(z))\right\}^{(1,1)}\right),
\end{equation}
where $\left\{\cdot\right\}^{(1,1)}$ denotes the upper-left $T_t\times T_t$ matrix block. Thus, to find the explicit expression of the Cauchy transform $\mc{G}_{\mb{\Xi},\mb{B}_{l}}(z)$ is amount to finding the operator-valued Cauchy transform $\mc{G}_{\mb{L}_{\mb{B}_l}}^{\mc{D}}(\mb{\Lambda}(z))$. 

\textbf{\textit{Step}~3: Calculate the operator-valued Cauchy transform}:
To describe the correlation of the channel coefficients, we define the one-sided correlation functions of $\widetilde{\mb{G}}_{kl}$ as 
\begin{align}
	\eta_{kl}(\widetilde{\mb{D}}) &= \mbb{E}[\widetilde{\mb{G}}_{kl}^\dagger\widetilde{\mb{D}} \widetilde{\mb{G}}_{kl}] = \frac{1}{L_k} {\mb{C}}_{kl}\mb{\Pi}_{kl}(\widetilde{\mb{D}}){\mb{C}}_{kl}^\dagger, \label{equ_eta} \\
	\widetilde{\eta}_{kl}({\mb{D}}) &= \mbb{E}[\widetilde{\mb{G}}_{kl}{\mb{D}} \widetilde{\mb{G}}_{kl}^\dagger] = \frac{1}{L_k}\mb{R}_{kl} \widetilde{\mb{\Pi}}_{kl}({\mb{D}}) \mb{R}_{kl}^\dagger, \label{equ_tildeeta}
\end{align}
where $\widetilde{\mb{D}}\in \mbb{C}^{R\times R}$ and $\mb{D}\in \mbb{C}^{L_k\times L_k}$ are arbitrary Hermitian matrices. The diagonal entries in the diagonal matrix $\mb{\Pi}_{kl}(\widetilde{\mb{D}})\in\mbb{C}^{L_k\times L_k}$ and the diagonal matrix $\widetilde{\mb{\Pi}}_{kl}({\mb{D}})\in\mbb{C}^{R\times R}$ are $\left[\mb{\Pi}_{kl}(\widetilde{\mb{D}})\right]_{ii} = \sum_{j = 1}^{R} \left([\mb{N}_{kl}]_{ji}\right)^2 \left[\mb{R}_{kl}^\dagger \widetilde{\mb{D}} \mb{R}_{kl}\right]_{jj}$ and $\left[\widetilde{\mb{\Pi}}_{kl}({\mb{D}})\right]_{ii} = \sum_{j=1}^{L_k} \left([\mb{N}_{kl}]_{ij}\right)^2 \left[ {\mb{C}}_{kl}^\dagger {\mb{D}} {\mb{C}}_{kl} \right]_{jj}$, respectively.

Similarly, for $1\leq n\leq N$, $1\leq l\leq L$ and  $1\leq k\leq K$ the two parameterized one-sided correlation functions of the matrices $\widetilde{\mb{F}}_{0,nl}$ and $\widetilde{\mb{F}}_{nk}$ are respectively given by
\begin{align}
	\zeta_{0,nl}(\mb{Z}_0) &= \mbb{E}[\widetilde{\mb{F}}_{0,nl}^\dagger\mb{Z}_0\widetilde{\mb{F}}_{0,nl}] = \frac{1}{T} \mb{P}_{0,nl}\mb{\Sigma}_{0,nl}(\mb{Z}_0)\mb{P}_{0,nl}^\dagger, \notag\\
	\widetilde{\zeta}_{0,nl}(\widetilde{\mb{Z}}_{0n}) &= \mbb{E}[\widetilde{\mb{F}}_{0,nl}\widetilde{\mb{Z}}_{0n}\widetilde{\mb{F}}_{0,nl}^\dagger] = \frac{1}{T}\mb{T}_{0,nl} \widetilde{\mb{\Sigma}}_{0,nl}(\widetilde{\mb{Z}}_{0n}) \mb{T}_{0,nl}^\dagger, \notag\\
	\zeta_{nk}(\mb{Z}) &= \mbb{E}[\widetilde{\mb{F}}_{nk}^\dagger\mb{Z}\widetilde{\mb{F}}_{nk}] = \frac{1}{T} \mb{P}_{nk}\mb{\Sigma}_{nk}(\mb{Z})\mb{P}_{nk}^\dagger, \notag\\
	\widetilde{\zeta}_{nk}(\widetilde{\mb{Z}}_n) &= \mbb{E}[\widetilde{\mb{F}}_{nk}\widetilde{\mb{Z}}_n\widetilde{\mb{F}}_{nk}^\dagger] = \frac{1}{T}\mb{T}_{nk} \widetilde{\mb{\Sigma}}_{nk}(\widetilde{\mb{Z}}_n) \mb{T}_{nk}^\dagger,\notag
\end{align}
where $\mb{Z}_0\in\mbb{C}^{R\times R}$, $\widetilde{\mb{Z}}_{0n}\in\mbb{C}^{T\times T}$, $\mb{Z}\in\mbb{C}^{L_k\times L_k}$ and $\widetilde{\mb{Z}}_n\in\mbb{C}^{T\times T}$ are arbitrary Hermitian matrices. The diagonal matrices $\mb{\Sigma}_{0,nl}(\mb{Z})\in\mbb{C}^{T\times T}$, $\mb{\Sigma}_{nk}(\mb{Z})\in\mbb{C}^{T\times T}$,  $\widetilde{\mb{\Sigma}}_{nk}(\widetilde{\mb{Z}}_n)\in\mbb{C}^{L_k\times L_k}$ and $\widetilde{\mb{\Sigma}}_{0,nl}(\widetilde{\mb{Z}}_{0n})\in\mbb{C}^{R\times R}$ respectively contain the diagonal entries $\left[\mb{\Sigma}_{0,nl}(\mb{Z}_0)\right]_{i,i} = \sum_{j = 1}^{L_k} \left([\mb{M}_{0,nl}]_{j,i}\right)^2 \left[\mb{T}_{0,nl}^\dagger \mb{Z}_0 \mb{T}_{0,nl}\right]_{j,j}$, $\left[\widetilde{\mb{\Sigma}}_{0,nl}(\widetilde{\mb{Z}}_{0n})\right]_{i,i} = \sum_{j=1}^{T} \left([\mb{M}_{0,nl}]_{i,j}\right)^2 \left[ \mb{P}_{0,nl}^\dagger \widetilde{\mb{Z}}_{0n} \mb{P}_{0,nl} \right]_{j,j}$, $\left[\mb{\Sigma}_{nk}(\mb{Z})\right]_{i,i} = \sum_{j = 1}^{L_k} \left([\mb{M}_{nk}]_{j,i}\right)^2 \left[\mb{T}_{nk}^\dagger \mb{Z} \mb{T}_{nk}\right]_{j,j}$ and $\left[\widetilde{\mb{\Sigma}}_{nk}(\widetilde{\mb{Z}}_n)\right]_{i,i} = \sum_{j=1}^{T} \left([\mb{M}_{nk}]_{i,j}\right)^2 \left[ \mb{P}_{nk}^\dagger \widetilde{\mb{Z}}_n \mb{P}_{nk} \right]_{j,j}$.

Further, the parameterized one-sided correlation matrices of $\mb{\wt F}_{0l}$ and $\mb{\wt F}_k$ are defined as 
\begin{align}
	\zeta_{0l}(\mb{Z}_0) &= \!\mbb{E}\{\widetilde{\mb{F}}_{0l}^\dagger\mb{\mb{Z}_0}\widetilde{\mb{F}}_{0l}\}\notag\\ &=\text{blkdiag}\left(\zeta_{0,1l}(\mb{\mb{Z}_0}),\zeta_{0,2l}(\mb{\mb{Z}_0}),\cdots,\zeta_{0,Nl}(\mb{\mb{Z}_0})\right), \label{equ_zeta}	\\
	\widetilde{\zeta}_{0l}({\mb{\wt Z}_0}) &= \mbb{E}\{\widetilde{\mb{F}}_{0l}{\mb{\wt Z}_0}\widetilde{\mb{F}}_{0l}^\dagger\} = \sum\limits_{n = 1}^N 	\widetilde{\zeta}_{0,nl}(\mb{\wt Z}_{0n}) ,\\
	\zeta_{k}(\mb{Z}) &= \!\mbb{E}\{\widetilde{\mb{F}}_{k}^\dagger\mb{Z}\widetilde{\mb{F}}_{k}\} \notag\\ 
	&=\text{blkdiag}\left(\zeta_{1k}(\mb{Z}),\zeta_{2k}(\mb{Z}),\cdots,\zeta_{Nk}(\mb{Z})\right), 	\\
	\widetilde{\zeta}_{k}(\widetilde{\mb{Z}}) &= \mbb{E}\{\widetilde{\mb{F}}_{k}\widetilde{\mb{Z}}\widetilde{\mb{F}}_{k}^\dagger\} = \sum\limits_{n = 1}^N 	\widetilde{\zeta}_{nk}(\widetilde{\mb{Z}}_n) \label{equ_tildezeta}, 
\end{align}
where $\mb{\wt Z}_0=\operatorname{blkdiag}(\mb{\wt Z}_1,\dots,\mb{\wt Z}_N)$ and $\widetilde{\mb{Z}}=\operatorname{blkdiag}(\widetilde{\mb{Z}}_1,\dots,\widetilde{\mb{Z}}_N)$.

As proved in~\cite{zheng2023mutual}, $\mb{\wt L}_{\mb{B}_l}$ is an operator-valued semicircular variable and is free from the deterministic matrix $\mb{\ob L}_{\mb{B}_l}$ over $\mc{D}$. Thus, the operator-valued Cauchy transform $\mc{G}_{\mb{L}_{\mb{B}_l}}^{\mc{D}}(\mb{\Lambda}(z))$ can be obtained by using the subordination formula, which is given as follows
\begin{equation} \label{equ_subordination}
	\begin{aligned}
		\mc{G}_{\mb{L}_{\mb{B}_l}}^{\mc{D}}&(\mb{\Lambda}(z)) = \mc{G}_{\ob{\mb{L}}_{\mb{B}_l}}^{\mc{D}}\left(\mb{\Lambda}(z) - \mc{R}_{\widetilde{\mb{L}}_{\mb{B}_l}}^{\mc{D}}\left(\mc{G}_{\mb{L}_{\mb{B}_l}}^{\mc{D}}(\mb{\Lambda}_n(z))\right)\right) \\
		&= \mbb{E}_{\mc{D}}\Big\{\left(\mb{\Lambda}(z) - \mc{R}_{\widetilde{\mb{L}}_{\mb{B}_l}}^{\mc{D}}\left(\mc{G}_{\mb{L}_{\mb{B}_l}}^{\mc{D}}(\mb{\Lambda}_n(z))\right) - \ob{\mb{L}}_{\mb{B}_l}\right)^{-1}\Big\},
	\end{aligned}
\end{equation}
where $\mc{R}_{\widetilde{\mb{L}}_{\mb{B}_l}}^{\mc{D}}$ is the $R$-transform of $\mb{\wt L}_{\mb{B}_l}$. Since $\mb{\wt L}_{\mb{B}_l}$ is an operator-valued semicircular variable over $\mc D$, we have
\begin{equation} \label{equ_Rtrans}
	\begin{aligned}
		&\mc{R}_{\widetilde{\mb{L}}_{\mb{B}_l}}^{\mc{D}}\left(\mc{G}_{\mb{L}_{\mb{B}_l}}^{\mc{D}}(\mb{\Lambda}(z))\right) = \mbb{E}_{\mc D}\left[\mb{\wt L} \mc{G}_{\mb{L}_{\mb{B}_l}}^{\mc{D}}(\mb{\Lambda}(z)) \mb{\wt L}\right]=\\
		&\operatorname{blkdiag}
		\left\{
		\begin{array} {c}
			\mb{\wh W}^\dagger\left(\sum_{k=1}^{K}\zeta_{k}(\mc{G}_{\mb{X}_{4k}}(z))+\zeta_{0l}(\mc{G}_{\mb{X}_{40}}(z)) \right) \mb{\wh W} \\
			\mb{\wh \Theta}^\dagger{\bm \eta}_{l}(\mc{G}_{\mb{X}_3}(z))\mb{\wh \Theta} \\
			\sum_{k=1}^{K}\widetilde{\eta}_{kl}\left(\mb{\Theta}_k\mc{G}_{\mb{X}_{2k}}(z)\mb{\Theta}_k^\dagger\right) \\
			\bm{\wt \zeta}_l(\mc{G}_{\mb{X}_1}(z)) \\
		\end{array}
		\right\},
	\end{aligned} 
\end{equation}
the one-sided correlation matrices are
\begin{align}
	{\bm \eta}_{l}((\mc{G}_{\mb{X}_3}(z)) &=\operatorname{blkdiag}\left\{\mb{0}_R,\eta_{1l}(\mc{G}_{\mb{X}_3}(z)),\dots,\eta_{Kl}(\mc{G}_{\mb{X}_3}(z))\right\}, \notag \\
	\bm{\wt \zeta}_l(\mc{G}_{\mb{X}_1}(z)) &=\operatorname{blkdiag}\Big\{\wt \zeta_{0l}({\mb{\wh W}\mc{G}_{\mb{X}_1}(z)\mb{\wh W}}^\dagger),\notag\\
	&\wt \zeta_{1}({\mb{\wh W}\mc{G}_{\mb{X}_1}(z)\mb{\wh W}}^\dagger),\dots,\wt\zeta_{K}({\mb{\wh W}\mc{G}_{\mb{X}_1}(z)\mb{\wh W}}^\dagger)\Big\}. \notag
\end{align}

	By substituting (\ref{equ_lambdaz}), (\ref{equ_barL}) and (\ref{equ_Rtrans}) into the subordination formula~(\ref{equ_subordination}), we have
	\begin{equation} \label{equ_GLB_E}
		\begin{aligned}
			\mc{G}_{\mb{L}_{\mb{B}_l}}^{\mc{D}}(\mb{\Lambda}_n(z)) &= 
			\operatorname{blkdiag}\left\{\mc{G}_{\mb{X}_1}(z),\mc{G}_{\mb{X}_2}(z),\mc{G}_{\mb{X}_3}(z),\mc{G}_{\mb{X}_4}(z)\right\}\\
			=\mbb{E}_{\mc D}&\left[
			\begin{array} {cccc}
				\bm{\Upsilon}(z)& \mb{0} & \mb{0} &-\mb{\wh W}^\dagger \mb{\ob F}_l^\dagger \\
				\mb{0} &\bm{\Gamma}(z)&-\mb{\wh \Theta}^\dagger\mb{\ob G}_l^{\dagger} & \mb{I}_{L_{AR}} \\
				\mb{0} &-\mb{\ob G}_l\mb{\wh \Theta}&\bm{\wt \Gamma}(z)& \mb{0} \\
				-\mb{\ob F}_l\mb{\wh W} & \mb{I}_{L_{AR}}& \mb{0} &\bm{\wt \Upsilon}(z) \\
			\end{array}
			\right]^{-1},
		\end{aligned}
	\end{equation}
	where the expressions of $\bm{\Upsilon}(z)$, $\bm{\Gamma}(z)$, $\bm{\wt \Gamma}(z)$ and $\bm{\wt \Upsilon}(z)$ are given in Proposition~\ref{prop_Cauchy}. By applying the matrix inversion lemma to (\ref{equ_GLB_E}), the expressions of $\mc{G}_{\mb{X}_1}(z)$, $\mc{G}_{\mb{X}_2}(z)$, $\mc{G}_{\mb{X}_3}(z)$ and $\mc{G}_{\mb{X}_4}(z)$ can be then obtained.
	
	\textbf{\textit{Step}~4: Calculate the Cauchy transform}: With the equivalence relation (\ref{equ_eq_G}), the Cauchy transform $\mc{G}_{\mb{\Xi},\mb{B}_{l}}(z)$ can be then obtained.

\section{Derivative of $\mc{\wt B}$ with Respect to $\mb{\Theta}_{k,l_k}^\star$} \label{appen_derivative_II}
Based on Proposition~\ref{prop_Cauchy}, we omit the notation $(z)$ and the subscript $l$ for convenience and define $\mb{I}^k_{l_k}$ as $\frac{d\bm{\wh \Theta}}{d\mb{\Theta}_{k,l_k}}$ and $\mb{I}_{l_k}$ as $\frac{d\bm{\Theta}_k}{d\mb{\Theta}_{k,l_k}}$. The derivative of $\mc{\wt B}$ with respect to $\mb{\Theta}_{k,l_k}^\star$ can be expressed as
\begin{equation} 
	\mc{\wt B}^\prime= -\mc{\wt B}\left(\bm{\Upsilon}^\prime + \mb{\wh W}^\dagger \mb{\ob F}^\dagger\mb{\Omega}^{-1}\mb{\Omega}^\prime\mb{\Omega}^{-1}\mb{\ob F}\mb{\wh W}\right)\mc{\wt B},
\end{equation}
where the matrix $\mb{\Omega}^\prime$ denotes the derivative of $\mb{\Omega}$.
\begin{figure*}[b] 
	\vspace{-0.4cm}
	\hrulefill  
	\begin{align}
		\mb{\Omega}^\prime&=\bm{\wt \Upsilon}^\prime+\left(\bm{\Gamma}-\mb{\wh \Theta}^\dagger\mb{\ob G}^{\dagger}\bm{\wt \Gamma}^{-1}\mb{\ob G}\mb{\wh \Theta}\right)^{-1}\left(\bm{\Gamma}^\prime-\mb{I}^k_{l_k}\mb{\ob G}^{\dagger}\bm{\wt \Gamma}^{-1}\mb{\ob G}\mb{\wh \Theta}+\mb{\wh \Theta}^\dagger\mb{\ob G}^{\dagger}\bm{\wt \Gamma}^{-1}\bm{\wt \Gamma}^\prime\bm{\wt \Gamma}^{-1}\mb{\ob G}\mb{\wh \Theta}\right)\left(\bm{\Gamma}-\mb{\wh \Theta}^\dagger\mb{\ob G}_l^{\dagger}\bm{\wt \Gamma}^{-1}\mb{\ob G}\mb{\wh \Theta}\right)^{-1}. \label{equ_varPi}
	\end{align}
\end{figure*}
The matrix-valued functions $\bm{\Upsilon}^\prime$, $\bm{\Gamma}^\prime$, $\bm{\wt \Gamma}^\prime$ and $\bm{\wt \Upsilon}^\prime$ satisfy the following fixed-point equations
\begin{align}
	\bm{\Upsilon}^\prime &=-\mb{\wh W}^\dagger\left(\sum_{k=1}^{K}\zeta_{k}(\mc{G}_{\mb{X}_{4k}}^\prime)+\zeta_{0l}(\mc{G}_{\mb{X}_{40}}^\prime) \right) \mb{\wh W}, \notag\\
	\bm{\Gamma}^\prime &= -\mb{I}^k_{l_k}\operatorname{blkdiag}\left\{\mb{0}_R,\eta_{1l}(\mc{G}_{\mb{X}_3}),\dots,\eta_{Kl}(\mc{G}_{\mb{X}_3})\right\}\mb{\wh \Theta}\notag \\
	&\qquad-\mb{\wh \Theta}^\dagger\operatorname{blkdiag}\left\{\mb{0}_R,\eta_{1l}(\mc{G}_{\mb{X}_3}^\prime),\dots,\eta_{Kl}(\mc{G}_{\mb{X}_3}^\prime)\right\}\mb{\wh \Theta}, \notag\\
	\bm{\wt \Gamma}^\prime &= -\widetilde{\eta}_{kl}\left(\mb{\Theta}_k\mc{G}_{\mb{X}_{2k}}\mb{I}_{l_k}\right)-\sum_{n=1}^{K}\widetilde{\eta}_{nl}\left(\mb{\Theta}_n\mc{G}^\prime_{\mb{X}_{2n}}\mb{\Theta}_n^\dagger\right), \notag\\
	\bm{\wt \Upsilon}^\prime &= -\operatorname{blkdiag}\Big\{\zeta_{0l}({\mb{\wh W}\mc{G}_{\mb{X}_1}^\prime\mb{\wh W}}^\dagger),\notag \\
	& \qquad\qquad\zeta_{1}({\mb{\wh W}\mc{G}_{\mb{X}_1}^\prime\mb{\wh W}}^\dagger),\dots,\zeta_{K}({\mb{\wh W}\mc{G}_{\mb{X}_1}^\prime\mb{\wh W}}^\dagger)\Big\}.\notag
\end{align}
where the equations $\mc{G}_{\mb{X}_1}^\prime$, $\mc{G}_{\mb{X}_{2k}}^\prime$, $\mc{G}_{\mb{X}_3}^\prime$ and $\mc{G}_{\mb{X}_{4k}}^\prime$ denote the derivatives of $\mc{G}_{\mb{X}_1}$, $\mc{G}_{\mb{X}_{2k}}$, $\mc{G}_{\mb{X}_3}$ and $\mc{G}_{\mb{X}_{4k}}$ with respect to $\mb{\Theta}_{k,l_k}^\star$, which can be directly obtained based on the matrix derivation and are omitted here for lack of space.

\end{appendices}

\bibliographystyle{IEEEtran}
\bibliography{RIS_WMMSE}

\end{document}